\documentclass{amsart}
\usepackage{amssymb}
\usepackage{amsmath}
\usepackage{amsthm}
\usepackage{mathrsfs}
\usepackage{colortbl}
\usepackage{fancybox}
\usepackage{graphicx}
\usepackage{ascmac}

\newtheorem{theorem}{Theorem}
\newtheorem{lemma}[theorem]{Lemma}
\newtheorem{assumption}[theorem]{Assumption}

\theoremstyle{definition}
\newtheorem{definition}[theorem]{Definition}

\newtheorem{proposition}[theorem]{Proposition}

\makeatletter

\@addtoreset{theorem}{section}
\makeatother

\makeatletter

\@addtoreset{equation}{section}
\makeatother

\begin{document}                                                 
\title[Energetic variational approaches for non-Newtonian fluid systems]{Energetic variational approaches for non-Newtonian fluid systems}



\author{Hajime Koba}
\address{ Graduate School of Engineering Science, Osaka University\\
1-3 Machikaneyamacho, Toyonaka, Osaka, 560-8531, Japan, Japan}
\email{iti@sigmath.es.osaka-u.ac.jp}
\thanks{This work was partly supported by the Japan Society for the Promotion of Science (JSPS) KAKENHI Grant Number JP15K17580.}

\author{Kazuki Sato}
\address{ Graduate School of Engineering Science, Osaka University\\
1-3 Machikaneyamacho, Toyonaka, Osaka, 560-8531, Japan, Japan}
\email{k-sato@sigmath.es.osaka-u.ac.jp}

\date{}
                      

\begin{abstract}
We consider the dominant equations for the motion of the non-Newtonian fluid in a domain from an energetic point of view. We apply our energetic variational approaches and the first law of thermodynamics to derive the generalized compressible non-Newtonian fluid system. We also derive the generalized incompressible non-Newtonian fluid system by using our energetic variational approaches.
\end{abstract}

\maketitle


\section{Introduction}\label{sect1}

The viscosities of yogurt, mayonnaise, and ketchup are higher than the one of the water. Of course, the motion of the particles in such fluids is different from the one of the water. We call such a fluid a non-Newtonian fluid. In this paper, we focus on several energies dissipation due to the viscosities of the non-Newtonian fluid to study the dominant equations for the motion of the fluid. We apply our energetic variational approaches to make a mathematical model of the non-Newtonian fluid flow. Although the system is abstract, our results make it possible to give some previous models of non-Newtonian fluid flow to a mathematical validity. Since mathematical derivation of the pressure of compressible fluid is different from the one of the pressure of incompressible fluid, this paper deals with both compressible and incompressible non-Newtonian fluids.

We first introduce fundamental notation. Let $t \geq 0$ be the time variable, and $x = { }^t (x_1 , x_2, x_3 )$, $\xi = { }^t ( \xi_1 , \xi_2 , \xi_3 ) \in \mathbb{R}^3$ the space variables. Let $\Omega (t) (= \{ \Omega (t) \}_{0 \leq t < T})$ be a bounded $C^2$-domain in $\mathbb{R}^3$ depending on time $t \in [0, T )$ for some $T \in ( 0, \infty ]$. The notation $\rho = \rho (x,t)$, $v = v (x,t) = { }^t ( v_1 , v_2 , v_3 )$, $\sigma = \sigma (x,t)$, $\theta = \theta (x,t)$, $e = e ( x , t)$, $e_A = e_A (x,t)$, $h = h (x,t)$, $s = s (x,t)$, and $e_F = e_F (x,t)$ represent the \emph{density}, the \emph{velocity}, the \emph{pressure}, the \emph{temperature}, the \emph{internal energy}, the \emph{total energy}, the \emph{enthalpy}, the \emph{entropy}, and the \emph{free energy} of the fluid in the domain $\Omega ( t )$, respectively. The notation $F = F ( x , t ) = { }^t ( F_1 , F_2, F_3 )$ and $C = C (x,t)$ are the \emph{exterior force (gravity)} and the \emph{concentration of amount of substance} in the fluid in the domain $\Omega (t)$, respectively.

This paper has three purposes. The first one is to make an abstract model of compressible non-Newtonian fluid flow from an energetic point of view. More precisely, we apply both an energetic variational approach and the first law of thermodynamics to derive the following generalized \emph{compressible non-Newtonian fluid system}:
\begin{equation}\label{eq11}
\begin{cases}
D_t \rho + ( {\rm{div}} v ) \rho = 0& \text{ in } \Omega_T,\\
\rho D_t v + {\rm{grad}} \sigma = {\rm{div}} S ( v , 0 ) + \rho F& \text{ in } \Omega_T,\\
\rho D_t e + ({\rm{div}} v ) \sigma  ={\rm{div}} q_\theta + \tilde{e}_D& \text{ in } \Omega_T,\\
D_t C + ( {\rm{div}} v ) C =  {\rm{div}} q_C  & \text{ in } \Omega_T.
\end{cases}
\end{equation}
Here $D_t f = \partial_t f + ( v , \nabla ) f $, $( v , \nabla ) f = v_1 \partial_1 f + v_2 \partial_2 f + v_3 \partial_3 f$, ${\rm{div}} v = \nabla \cdot v$, ${\rm{grad}} f = \nabla f$, $\partial_t = \partial/{\partial t}$, $\partial_i = \partial/{\partial x_i}$, $\nabla = { }^t ( \partial_1 , \partial_2 , \partial_3 )$, 
\begin{equation*}
\Omega_T = \left\{ ( x , t) \in \mathbb{R}^4; { \ } ( x , t ) \in \bigcup_{0 < t <T} \{ \Omega (t) \times \{ t \} \} \right\},
\end{equation*} 
\begin{align*}
S (v , \sigma ) &= e_1' (|D_+ (v)|^2) D_+ (v) + e_2' ( |{\rm{div}} v |^2) ({\rm{div}} v) I_{3}  + e_3' (| D_- (v) |^2) D_- (v) - \sigma I_{3},\\
\tilde{e}_D & = e_1' (|D_+ (v)|^2) |D_+ (v)|^2 + e_2' ( |{\rm{div}} v |^2) |{\rm{div}} v|^2 + e_3' (| D_- (v) |^2) |D_- (v)|^2,\\
q_\theta &= e'_4 ( | {\rm{grad}} \theta |^2 ) {\rm{grad}} \theta ,\\
q_C &= e'_5 ( | {\rm{grad}} C |^2 ) {\rm{grad}} C,
\end{align*}
where $I_3 (=I_{3 \times 3})$ is the $3 \times 3$ identity matrix, $e_1, \cdots , e_5$ are $C^1$-functions, $e'_j = e'_j(r) = \frac{ d e_j }{d r} (r)$, $D_+ (v) = \{ (\nabla v ) + { }^t ( \nabla v ) \}/2$, $D_- ( v ) = \{ ( \nabla v) - { }^t ( \nabla v ) \}/2$, $| D_+ (v)|^2 = D_+ ( v) : D_+ (v)$, $ | D_- ( v) |^2 = D_- ( v ) : D_- ( v )$, and $|{\rm{grad}} f |^2 = \nabla f \cdot \nabla f$. We call $D_t$ the \emph{material derivative}, $\tilde{e}_D$ the \emph{density for the energy dissipation due to the viscosities}, $S (v , \sigma)$ the \emph{stress tensor}, $q_\theta$ the \emph{heat flux}, $q_C$ the \emph{general flux}, $D_+ ( v )$ the \emph{strain rate tensor}, and $D_- (v)$ the \emph{vorticity tensor}. Remark that
\begin{equation*}
2 | D_- ( v )|^2 = |{\rm{curl}} v |^2, \text{ where } {\rm{curl}} v = \nabla \times v .
\end{equation*}
In general, we consider the case when $e_3 \equiv 0$ to study the motion of the fluid. In this paper we deal with the case when $e_3 \neq 0$. Remark also that it is easy to check that
\begin{equation*}
{\rm{div}} S ( v , 0 )  = ( \mu_1 + \mu_3 ) \Delta v + (\mu_1 + \mu_2 - \mu_3 ) {\rm{grad}} ({\rm{div}} v )
\end{equation*}
if $e_j (r)= 2 \mu_j r$ for some $\mu_j \in \mathbb{R}$, where $\Delta = \partial_1^2 + \partial_2^2 + \partial_3^2$.

The second one is to study our compressible non-Newtonian fluid system \eqref{eq11}. More precisely, we investigate the conservative forms,  enthalpy, entropy, free energy, and conservation laws of the system \eqref{eq11}. In fact, we can write the system \eqref{eq11} as follows:
\begin{equation}\label{eq12}
\begin{cases}
\partial_t \rho + {\rm{div}} ( \rho v ) =0& \text{ in } \Omega_T,\\
\partial_t (\rho v ) + {\rm{div}} (\rho v \otimes v - S( v , \sigma ) ) = \rho F& \text{ in } \Omega_T,\\
\partial_t e_A + {\rm{div}} ( e_A v - q_\theta - S (v , \sigma ) v ) = \rho F \cdot v& \text{ in } \Omega_T,\\
\partial_t C + {\rm{div}} (C v - q_C) = 0& \text{ in } \Omega_T ,
\end{cases}
\end{equation}
where $e_A$ is the total energy define by $e_A = \rho |v|^2/2 + \rho e$. Under some assumptions, the enthalpy $h = h (x,t)$, the entropy $s = s( x,t)$, and the free energy $e_F = e_F (x,t)$ satisfy
\begin{equation}\label{eq13}
\begin{cases}
\partial_t ( \rho h ) + {\rm{div}} ( \rho h v - q_\theta ) = \tilde{e}_D + D_t \sigma & \text{ in } \Omega_T,\\
\partial_t ( \rho s ) + {\rm{div}} \left( \rho s v - \frac{q_\theta }{\theta} \right) = \frac{\tilde{e}_D}{\theta} + \frac{q_\theta \cdot {\rm{grad}}\theta}{\theta^2} & \text{ in } \Omega_T,
\end{cases}
\end{equation}
and
\begin{equation*}
\rho D_t e_F + s \rho D_t \theta - S (v , \sigma ) : ( D_+ (v) + D_- (v)) = - \tilde{e} _D . 
\end{equation*}
Moreover, the system \eqref{eq11} satisfies the following conservation of angular momentum:
\begin{equation*}
\int_{\Omega (t_2)} x \times \rho v { \ } d x = \int_{\Omega ( t_1 )} x \times \rho v { \ } d x + \int_{t_1}^{t_2} \int_{\Omega (\tau )} x \times \rho F { \ } d x d \tau .
\end{equation*}
See Theorem \ref{thm28} and Section \ref{sect6} for details.

The third one is to make an abstract model of incompressible non-Newtonian fluid flow from an energetic point of view. More precisely, we apply an energetic variational approach to derive the following generalized \emph{incompressible non-Newtonian fluid system}:
\begin{equation*}
\begin{cases}
D_t \rho = 0 & \text{ in } \Omega_T,\\
{\rm{div}} v = 0 & \text{ in } \Omega_T,\\
\rho D_t v + {\rm{grad}} \sigma = {\rm{div}} \{  e_1' ( | D_+ (v) |^2 ) D_+ (v) + e_3' ( | D_- ( v ) |^2 ) D_- ( v ) \} + \rho F & \text{ in } \Omega_T,\\
\rho D_t \theta =  {\rm{div}} \{ e'_4 ( | {\rm{grad}} \theta |^2 ) {\rm{grad}} \theta \} & \text{ in } \Omega_T,\\
D_t C = {\rm{div}} \{ e'_5 (|{\rm{grad}} C |^2) {\rm{grad}} C \}  & \text{ in } \Omega_T .
\end{cases}
\end{equation*}
Note that we often call our incompressible fluid system the Ostwald-de Waele model (\cite{Ost25}, \cite{Wae23}) when $e_3 \equiv 0$ and $e_1 (r) = \mu_1 r^p$ for some $\mu_1 > 0$ and $0 < p ( \neq 1) <  \infty$. Note also that we call our incompressible fluid system the incompressible Newtonian fluid system when $e_3 \equiv 0$ and $e_1 (r) = \mu_1 r$ for some $\mu_1 > 0$. Our incompressible fluid system includes several previous models of non-Newtonian fluid flow. See Bird-Armstrong-Hassager \cite{BAH87} for several non-Newtonian fluid models.

In this paper we apply our energetic variational approaches to make mathematical models of non-Newtonian fluid flow. Serrin \cite{Ser59} derived the incompressible and compressible fluid systems from mathematical and thermodynamical points of view. Bird \cite{B60} and Johnson \cite{J60} used their variational approaches to study non-Newtonian fluid. Hyon-Kwak-Liu \cite{HKL10} applied their energetic variational approaches, based on Strutt \cite{Str73} and Onsager (\cite{Ons31a}, \cite{Ons31b}), to study complex fluids. They combined the two systems derived from the least action principle and the maximum/minimum dissipation principle to derive hydrodynamic systems of complex fluids. This method is called an \emph{energetic variational approach}. Koba-Liu-Giga \cite{KLG17} and Koba \cite{K17} improved energetic variational approaches to derive incompressible and compressible fluid systems on an evolving surface, respectively. The ideas of this paper are based on Serrin \cite{Ser59}, Hyon-Kwak-Liu \cite{HKL10}, Koba-Giga-Liu \cite{KLG17}, Koba \cite{K17}, and this paper generalizes their energetic variational approaches to apply non-Newtonian fluid.

Let us explain our strategy for deriving our non-Newtonian fluid systems. We first set the following energy densities for non-Newtonian fluid.
\begin{assumption}[Energy densities for non-Newtonian fluid]\label{ass11}
\begin{multline*}
e_K  = \frac{1}{2} \rho | v |^2, { \ }e_D = \frac{1}{2} e_1 (| D_+ (v) |^2 ) + \frac{1}{2} e_2 (| {\rm{div}} v |^2 )  + \frac{1}{2} e_3 (| D_- (v)|^2),\\
e_{W}  = ({\rm{div}} v ) \sigma + \rho F \cdot v,{ \ }e_{TD}  = \frac{1}{2} e_4 (|{\rm{grad}} \theta |^2),{ \ }e_{GD} = \frac{1}{2} e_5 (|{\rm{grad}} C|^2).
\end{multline*}
Here $e_1,e_2,e_3,e_4,e_5$ are $C^1$-functions.
\end{assumption}
\noindent We call $e_K$ the \emph{kinetic energy}, $e_D$ the \emph{energy density for the energy dissipation due to the viscosities}, $e_{W}$ the \emph{power density for the work done by the pressure and exterior force}, $e_{TD}$ the \emph{energy density for the energy dissipation due to thermal diffusion}, and $e_{GD}$ the \emph{energy density for the energy dissipation due to general diffusion}. Note that $e_D \neq \tilde{e}_D$ in general, where $\tilde{e}_D$ is the density for the energy dissipation due to the viscosities. See subsection \ref{subsec61} for details.

Secondly, we consider a mathematical validity of our energy densities for non-Newtonian fluid. We apply a flow map and the Riemannian metric induced by the flow map to study the invariance of our energy densities (see Section \ref{sect3}). 

Thirdly, we derive forces from a variations of some energies based on our energy densities (see Sections \ref{sect4} and \ref{sect5}). 

Finally, we make use of several forces and our energetic variational approaches to make a mathematical model of non-Newtonian fluid flow (see Section \ref{sect6}).

This paper provides mathematical models of non-Newtonian fluid flow. We refer the reader to Serrin \cite{Ser59} for mathematical derivations of fluid systems and Gurtin-Fried-Anand \cite{GFA10} for physical laws such as the laws of thermodynamics.

\section{Main Results}\label{sect2}

In this section, we first state some conventions of this paper. Then we introduce a flow map in a domain and the Riemannian metric induced by the flow map, which are two key tools of our energetic variational approaches. Finally, we state the main results of this paper.

Let us explain the conventions used in this paper. We use the italic characters $i,j,k,\ell,i' j',k', \ell'$ as $1,2,3$. Moreover, we often use the following Einstein summation convention: $a_{i j}b_j = \sum_{j=1}^3 a_{i j} b_j$, $a_{i j} b^{i j} = \sum_{i, j=1}^3 a_{i j}b^{i j}$.

Let $\mathcal{X}$ be a set. The symbol $M_{p \times q} ( \mathcal{X}) $ denotes the set of all $p \times q$ matrices whose component belonging to $\mathcal{X}$, that is, $M \in M_{p \times q} ( \mathcal{X}) $ if and only if $[M]_{\mathfrak{i} \mathfrak{j}} \in \mathcal{X}$ $(\mathfrak{i}=1,2,\ldots, p, { \ }\mathfrak{j}=1,2,\ldots, q)$, where $[M]_{\mathfrak{i} \mathfrak{j}}$ denotes the $(\mathfrak{i},\mathfrak{j})$-th component of the matrix $M$. For $M_1, M_2 \in M_{p \times q} (\mathbb{R})$, $M_1 : M_2 := \sum_{\mathfrak{i}=1}^p \sum_{\mathfrak{j} =1}^q [M_1]_{\mathfrak{i} \mathfrak{j}} [M_2]_{\mathfrak{i} \mathfrak{j}}$. In particular, we set $|M_1|^2 = M_1: M_1$. Remark that we can write $M = ([M]_{\mathfrak{i} \mathfrak{j}})_{p \times q}$ and $I_3 = I_{3 \times 3} = ( \delta_{i j})_{3 \times 3} $, where $\delta_{i j }$ is the Kronecker delta.

Let $\Omega (t) ( = \{ \Omega (t) \}_{0 \leq t < T})$ be a $C^2$-bounded domain in $\mathbb{R}^3$ depending on time $t \in [0, T)$ for some $T \in ( 0 , \infty ]$. The symbol $C_0^\infty ( \Omega (t))$ denotes the set of all smooth functions whose support are in $\Omega (t)$. Throughout this paper we assume that $\rho = \rho (x,t)$, $v = v ( x ,t )= { }^t (v_1, v_2, v_3)$, $\sigma = \sigma (x,t)$, $e = e (x,t)$, $\theta = \theta (x,t)$, $F = F (x, t) = { }^t ( F_1 , F_2, F_3 )$, $C = C ( x ,t)$, $h = h (x,t)$, $s = s (x , t)$ , $e_F = e_F (x,t) $ are smooth functions, where a smooth function means a $C^1$ or $C^2$-function in this paper.

We introduce two key tools of this paper. We say that $U (t) \subset \Omega ( t )$ is \emph{flowed by the velocity field} $V = V (x,t) = { }^t ( V_1 (x,t) , V_2 (x,t) , V_3 (x,t) )$ if there exists a smooth function $\tilde{x} = \tilde{x} (\xi , t ) = { }^t (\tilde{x}_1 (\xi, t), \tilde{x}_2 (\xi ,t), \tilde{x}_3 (\xi ,t ) ) $ such that for every $\xi \in \Omega (0)$,
\begin{equation*}
\begin{cases}
\frac{d \tilde{x}}{d t} (\xi , t) = V (\tilde{x} (\xi , t), t ), { \ \ \ }t \in (0,T),\\
\tilde{x}(\xi , 0) = \xi ,
\end{cases}
\end{equation*}
and $U ( t )$ is expressed by
\begin{equation*}
U ( t )= \{ x= { }^t (x_1,x_2,x_3) \in \mathbb{R}^3;{ \ }x = \tilde{x} (\xi , t ) , { \ }\xi \in U_0, { \ } U_0 \subset \Omega (0) \}.
\end{equation*}
The mapping $\xi \mapsto \tilde{x} (\xi ,t)$ is called a \emph{flow map} in $\Omega (t)$, the mapping $t \mapsto \tilde{x} (\xi , t)$ is called an \emph{orbit} starting from $\xi$, and $V = V (x,t)$ is called the velocity determined by the flow map $\tilde{x} = \tilde{x} ( \xi , t)$. For simplicity we call $\tilde{x} ( \xi , t )$ a flow map. We assume that $\tilde{x} ( \cdot , t) : \Omega (0) \to \Omega (t)$ is bijective for each $0 < t < T$. For the flow map $\tilde{x} = \tilde{x} (\xi , t)$ in $\Omega (t)$,
\begin{equation*}
g_i =g_i (\xi,t) := \frac{\partial \tilde{x}}{\partial \xi_i} ={ }^t \left(\frac{\partial \tilde{x}_1}{\partial \xi_i},\frac{\partial \tilde{x}_2}{\partial \xi_i},\frac{\partial \tilde{x}_3}{\partial \xi_i} \right).
\end{equation*}
Set
\begin{align*}
& g_{i j} = g_{ i j} ( \xi , t) := g_i \cdot g_j = \frac{\partial \tilde{x}_\ell}{\partial \xi_i} \frac{\partial \tilde{x}_\ell}{\partial \xi_j} = \sum_{\ell =1}^3 \frac{\partial \tilde{x}_\ell}{\partial \xi_i} \frac{\partial \tilde{x}_\ell}{\partial \xi_j}  ,\\
& ( g^{i j} )_{3 \times 3} := ((g_{i j})_{3 \times 3} )^{-1}, \text{that is, } (g^{i j})_{3 \times 3}( g_{i j})_{3 \times 3} = I_{3 \times 3}, \\
& g^i := g^{i j}g_j = g^{i 1}g_1 + g^{i 2}g_2 + g^{i 3} g_{3},\\
& J := J (\xi ,t ) = \sqrt{ {\rm{det}} (g_{i j})_{3 \times 3} }.
\end{align*}
In this paper, the notation $g_i$, $g^j$, $g_{i j}$, $g^{i j}$, $\sqrt{{ \rm{det}} ( g_{i j})_{3 \times 3} }$ are collectively called the Riemannian metrics induced by the flow map $\tilde{x} = \tilde{x} ( \xi , t)$. See Section \ref{sect3} for some properties of the flow maps and Riemannian metrics. Moreover,
\begin{equation*}
\acute{g}_i := \frac{d}{d t} g_i  = \frac{\partial V}{\partial \xi_i} = { }^t \left( \frac{\partial V_1}{\partial \xi_i} , \frac{\partial V_2}{\partial \xi_i} ,  \frac{\partial V_3}{\partial \xi_i} \right), { \ } \acute{g}_{i j} := \frac{d}{d t} (g_{i j}) = \acute{g}_i \cdot g_j + g_i \cdot \acute{g}_j .
\end{equation*}

Now we state the main results of this paper. We begin by studying our energy densities for non-Newtonian fluid (see Assumption \ref{ass11}). 
\begin{theorem}[Representation of energy densities]\label{thm21}
Assume that $\Omega (t)$ is flowed by the smooth velocity fields $V = V (x,t) = { }^t ( V_1 , V_2 , V_3)$. Set
\begin{align*}
\mathcal{K}(e_{W_1}) = \mathcal{K}(e_{W_1}) (\xi , t)= & \frac{1}{2} \acute{g}_{i j} g^{i j} \sigma (\tilde{x} (\xi,t) , t ),\\
\mathcal{K} (e_{D_1}) =\mathcal{K} (e_{D_1}) ( \xi ,t ) =& \frac{1}{2} e_1 \left(  \frac{1}{4} ( \acute{g}_{i j}\acute{g}_{k \ell} g^{i k}g^{j \ell} )  \right),\\
\mathcal{K} (e_{D_2}) =\mathcal{K} (e_{D_2})  (\xi , t) = & \frac{1}{2} e_2 \left( \frac{1}{4}  (\acute{g}_{i j} \acute{g}_{k \ell} g^{i j} g^{k \ell}) \right) ,\\
\mathcal{K} (e_{D_3}) =\mathcal{K} (e_{D_3}) (\xi , t) =& \frac{1}{2} e_3 \left(  ( \acute{g}_i \cdot \acute{g}_j)( g^i \cdot g^j) - \frac{1}{4} ( \acute{g}_{i j}\acute{g}_{k \ell} g^{i k}g^{j \ell} )  \right)  ,\\
\mathcal{K}(e_{D_4}) =\mathcal{K}(e_{D_4}) ( \xi , t) =& \frac{1}{2} e_4 \left( g^{i j}  \frac{\partial \theta }{\partial \xi_i} \frac{\partial \theta }{\partial \xi_j} \right),\\
 \mathcal{K}(e_{D_5}) = \mathcal{K}(e_{D_5}) ( \xi , t ) =& \frac{1}{2} e_5 \left( g^{i j}  \frac{\partial C }{\partial \xi_i} \frac{\partial C }{\partial \xi_j} \right),
\end{align*}
where $\tilde{x} = \tilde{x} ( \xi , t) $ is a flow map in $\Omega (t)$. Then
\begin{align*}
\int_{\Omega (t)} ({\rm{div}} V ) \sigma { \ }d x= \int_{\Omega (0)}\mathcal{K}(e_{W_1})J (\xi,t){ \ }d \xi,\\
\int_{\Omega (t)} \frac{1}{2} e_1 (|D_+ (V)|^2) { \ }d x= \int_{\Omega (0)} \mathcal{K}(e_{D_1}) J (\xi,t){ \ }d \xi,\\
\int_{\Omega (t)} \frac{1}{2} e_2 ( | {\rm{div}} V |^2 ) { \ }d x= \int_{\Omega (0)} \mathcal{K}(e_{D_2}) J (\xi,t){ \ }d \xi,\\
\int_{\Omega (t)} \frac{1}{2} e_3 ( | D_- (V) |^2 ) { \ }d x= \int_{\Omega (0)} \mathcal{K}(e_{D_3}) J (\xi,t){ \ }d \xi,\\
\int_{\Omega (t)} \frac{1}{2} e_4 ( | {\rm{grad}} \theta |^2  )  { \ }d x= \int_{\Omega (0)}\mathcal{K}(e_{D_4}) J (\xi,t){ \ }d \xi,\\
\int_{\Omega (t)} \frac{1}{2} e_5 ( | {\rm{grad}} C |^2  )  { \ }d x= \int_{\Omega (0)}\mathcal{K}(e_{D_5}) J (\xi,t){ \ }d \xi.
\end{align*}
\end{theorem}
\noindent See Section \ref{sect3} for the representation of another energy density. Theorem \ref{thm21} gives a mathematical validity of our energy densities. More precisely, we see the invariance of our energy densities from Theorem \ref{thm21} and \cite[Section 3]{K17}. The paper \cite{K17} studied compressible fluid systems on an evolving surface from an energetic point of view.

Next we consider the viscous, pressure, and diffusion terms of our systems. In this paper we derive these terms by applying our energy densities. For smooth functions $V = V (x,t) = { }^t (V_1 , V_2 , V_3)$ and $f = f ( x ,t )$,
\begin{align*}
E_{D} [V] (t) & : = - \int_{ \Omega (t)} \frac{1}{2} \left\{ e_1 ( |D_+ ( V ) |^2 ) + e_2 (|{\rm{div}} V |^2) + e_3 (| D_- (V) |^2 ) \right\} { \ }d x,\\
E_{W} [V] (t) & := \int_{\Omega (t)} \{ ({\rm{div}} V ) \sigma + \rho F \cdot V \}{ \ } d x,\\
E_{TD} [f] (t)  & := - \int_{\Omega (t)} \frac{1}{2} e_4 (| {\rm{grad}} f |^2) { \ } d x,\\
E_{GD} [f] (t)  & := - \int_{\Omega (t)} \frac{1}{2} e_5 (| {\rm{grad}} f |^2 ) { \ } d x .
\end{align*}
Moreover, we set
\begin{equation*}
E_{D + W} [V] (t) =  E_{D} [V] (t) + E_{W}[V] (t).
\end{equation*}

Let us derive the viscous and pressure terms of our fluid systems and the diffusion terms of our heat and diffusion systems.
\begin{theorem}[Derivation of viscous and pressure terms of compressible fluid systems]\label{thm22}
Fix $t \in (0,T)$. Assume that for every $\varphi \in [ C_0^\infty ( \Omega (t)) ]^3$,
\begin{equation*}
\frac{d}{d \varepsilon} \bigg|_{\varepsilon = 0} E_{D + W} [ v + \varepsilon \varphi ] (t)= 0.
\end{equation*}
Then $( \rho, v , \sigma , F )$ fulfills
\begin{multline*}
{\rm{div}} \{ e_1' ( |D_+ ( v) |^2 ) D_+ (v) + e_2' ( | {\rm{div}} v |^2 ) ({\rm{div}} v  )I_{3 \times 3} \\
+ e_3' ( | D_- (v) |^2 ) D_- (v) \} - {\rm{grad}} \sigma + \rho F = 0.
\end{multline*}
\end{theorem}
\begin{theorem}[Derivation of viscous and pressure terms of incompressible fluid systems]\label{thm23}
Fix $t \in (0, T)$. Suppose that ${\rm{div}} v = 0$. Assume that for every $\varphi \in [C_0^\infty ( \Omega (t))]^3$ satisfying ${\rm{div}} \varphi = 0$,
\begin{equation*}
\frac{d}{d \varepsilon} \bigg|_{\varepsilon = 0} E_{D + W} [ v + \varepsilon \varphi ] (t)= 0.
\end{equation*}
Then there is a function $\sigma \in C^1 ( \Omega (t) )$ such that
\begin{equation*}
{\rm{div}} \{ e_1' ( |D_+ ( v) |^2 ) D_+ (v) + e_3' ( | D_- (v) |^2) D_- (v) \} + \rho F = {\rm{grad}} \sigma .
\end{equation*}
\end{theorem}

\begin{theorem}[Variations of dissipation energies]\label{thm24}Fix $t \in (0, T)$. Then for every $\phi \in C_0^\infty ( \Omega (t))$,
\begin{align*}
\frac{d}{d \varepsilon} \bigg|_{\varepsilon = 0} E_{TD} [ \theta + \varepsilon \phi ] (t) & = \int_{\Omega (t)} {\rm{div}} \{ e_4' ( | {\rm{grad}} \theta |^2 ) {\rm{grad} \theta} \} \phi { \ } d x,\\
\frac{d}{d \varepsilon} \bigg|_{\varepsilon = 0} E_{GD} [ C + \varepsilon \phi ] (t) & = \int_{\Omega (t)} {\rm{div}} \{ e_5' ( | {\rm{grad}} C |^2 ) {\rm{grad} C} \} \phi { \ } d x. 
\end{align*}
\end{theorem}
\noindent From Theorems \ref{thm22}-\ref{thm24} we obtain forces derived from a variation of energies base on our energy densities. See the Appendix $(\mathrm{I})$ for another derivation of strain rate tensors and fluxes from our energy densities.

Remark: Mathematical derivation of the pressure of compressible fluid is different form the one of the pressure of incompressible fluid. In this paper we make use of the power density $e_W$ to derive the pressure of compressible fluid. On the other hand, we obtain the pressure of incompressible fluid to apply the following proposition.
\begin{proposition}[Temam \cite{Tem77}, Sohr \cite{Soh01}]\label{prop25}
Let $\Omega$ be a bounded $C^2$-domain in $\mathbb{R}^3$. Set
\begin{equation*}
G_2 ( \Omega ) = \left\{ f \in [L^2 (\Omega ) ]^3; { \ } \int_{ \Omega } f (x) \cdot \varphi (x) { \ } d x = 0 \text{ for } \varphi \in [ C_0^\infty ( \Omega )]^3 \text{ with } {\rm{div}} \varphi =0 \right\} .
\end{equation*}
Then $f \in G_2 ( \Omega)$ if and only if there is $\mathfrak{p} \in W^{1,2} ( \Omega )$ such that $f = \nabla \mathfrak{p}$. Moreover, $f$ is continuous, then $\mathfrak{p}$ is $C^1$-function.
\end{proposition}
\noindent Note that we can obtain the regularity of the function $\mathfrak{p}$ from the elliptic regularity theory. Note also that this paper does not characterize the pressure of compressible non-Newtonian fluid. See Appendix $(\mathrm{II})$ for the derivation of the compressible barotropic fluid system.

Now we consider a variation of the action integral determined by the kinetic energy with respect to the flow maps. To this end, we introduce a variation $\tilde{x}^\varepsilon (\xi,t)$ of a flow map $\tilde{x}(\xi,t)$ and the velocity $v^\varepsilon$ determined by the flow map $\tilde{x}^\varepsilon$. Let $\tilde{x} (\xi, t)$ be a flow map in $\Omega (t)$, and let $v$ be the velocity determined by the flow map $\tilde{x} (\xi ,t)$ in $\Omega ( t )$, i.e. for $\xi \in \Omega (0)$ and $0<t<T$,
\begin{equation*}
\begin{cases}
v = v (x,t)= { }^t ( v_1 (x,t) , v_2 ( x,t) , v_3 (x,t) ),\\ 
\tilde{x} = \tilde{x} ( \xi ,t ) = { }^t ( \tilde{x}_1 (\xi,t ) , \tilde{x}_2 (\xi,t ) , \tilde{x}_3 (\xi,t ) ),\\
\frac{d \tilde{x}}{d t} (\xi , t) = v ( \tilde{x} (\xi , t ) , t ),\\
\tilde{x} (\xi , 0) = \xi .
\end{cases}
\end{equation*}
Write
\begin{align*}
\Omega (t) & := \{ x = { }^t (x_1 , x_2 , x_3 ) \in \mathbb{R}^3; { \ } x = \tilde{x} (\xi , t ) , { \ }\xi  \in \Omega (0) \} ,\\
\Omega_T & := \bigg\{ (x,t) \in \mathbb{R}^4;{ \ } (x,t ) \in \bigcup_{0<t <T} \{ \Omega (t) \times \{ t \} \} \bigg\}.
\end{align*}
\noindent For $-1 < \varepsilon <1$, let $\Omega^\varepsilon (t) ( = \{ \Omega^\varepsilon (t) \}_{0 \leq t < T})$ be a domain in $\mathbb{R}^3$ depending on time $t \in [0, T )$. We say that $\Omega^\varepsilon (t)$ is a variation of $\Omega (t)$ if $\Omega^\varepsilon (0) = \Omega (0)$ and $\Omega^{\varepsilon}(t)|_{\varepsilon = 0} = \Omega (t) $. Set
\begin{equation*}
\Omega_T^\varepsilon := \bigg\{ (x,t) \in \mathbb{R}^4;{ \ } (x,t ) \in \bigcup_{0<t <T} \{ \Omega^\varepsilon (t) \times \{ t \} \} \bigg\}.
\end{equation*}
Let $\tilde{x}^\varepsilon (\xi,t)$ be a flow map in $\Omega^\varepsilon (t)$, and $v^\varepsilon$ be the velocity determined by the flow map $\tilde{x}^\varepsilon$, i.e. for $\xi \in \Omega (0)$ and $0<t<T$,
\begin{equation*}
\begin{cases}
v^\varepsilon = v^\varepsilon (x,t)= { }^t ( v^\varepsilon_1 (x,t) , v^\varepsilon_2 ( x,t) , v^\varepsilon_3 (x,t) ),\\ 
\tilde{x}^\varepsilon = \tilde{x}^\varepsilon ( \xi ,t ) = { }^t ( \tilde{x}^\varepsilon_1 (\xi,t ) , \tilde{x}^\varepsilon_2 (\xi,t ) , \tilde{x}^\varepsilon_3 (\xi,t ) ),\\
\frac{d \tilde{x}^\varepsilon}{d t} (\xi , t) = v^\varepsilon (\tilde{x}^\varepsilon(\xi , t ) , t ),\\
\tilde{x}^\varepsilon (\xi , 0) = \xi .
\end{cases}
\end{equation*}
We say that $( \tilde{x}^\varepsilon (\xi ,t) , \Omega_T^\varepsilon)$ is a variation of $( \tilde{x} (\xi ,t ) , \Omega_T)$ if $\tilde{x}^\varepsilon (\xi , t)$ is smooth as a function of $( \varepsilon , \xi , t ) \in (-1 ,1) \times \Omega (0) \times [0,T)$ and $\tilde{x}^\varepsilon ( \xi ,t ) |_{ \varepsilon = 0} = \tilde{x} ( \xi , t)$. Assume that $\Omega^\varepsilon (t)$ is expressed by
\begin{equation*}
\Omega^\varepsilon (t) = \{ x = { }^t (x_1 , x_2 , x_3 ) \in \mathbb{R}^3; { \ } x = \tilde{x}^\varepsilon (\xi , t ) , { \ }\xi  \in \Omega (0) \} .
\end{equation*}

Now we consider the densities of the fluid in the domain $\Omega (t)$ and a variation of $\Omega (t)$. Applying the Reynold transport theorem, we have
\begin{proposition}[Continuity equation in domains]\label{prop26}{ \ }\\
$(\mathrm{i})$ Assume that for each $0 < t <T$ and every $U ( t ) \subset \Omega (t)$ flowed by the velocity $v$,
\begin{equation*}
\frac{d}{d t} \int_{U (t)} \rho ( x , t) { \ }d x = 0.
\end{equation*}
Then
\begin{equation*}
\partial_t \rho + ( v , \nabla ) \rho + ({\rm{div}} v ) \rho = 0 \text{ in } \Omega_T. 
\end{equation*}
$(\mathrm{ii})$ Assume that for each $0< t <T$ and every $U^\varepsilon ( t ) \subset \Omega^\varepsilon (t)$ flowed by the velocity $v^\varepsilon$,
\begin{equation*}
\frac{d}{d t} \int_{U^\varepsilon (t)} \rho^\varepsilon ( x , t ) { \ }d x = 0.
\end{equation*}
Then
\begin{equation*}
\partial_t \rho^\varepsilon + (v^\varepsilon , \nabla ) \rho^\varepsilon + ({\rm{div}} v^\varepsilon ) \rho^\varepsilon = 0 \text{ in } \Omega_T^\varepsilon.  
\end{equation*}
\end{proposition}
\noindent To study a variation of the action integral determined by the kinetic energy, we give the proof of the assertion $(\mathrm{ii})$ of Proposition \ref{prop26} in Section \ref{sect3}.

Let us consider a variation of our action integral based on the kinetic energy $e_K$. Let $\rho_0 \in C^1 (\Omega (0))$. Assume that $\rho^\varepsilon$ and $\rho$ satisfy
\begin{equation*}
\begin{cases}
\partial_t \rho + ( v , \nabla ) \rho + ({\rm{div}} v ) \rho = 0 & \text{ in } \Omega_T,\\
\rho |_{t=0} = \rho_0 & \text{ in } \Omega (0),
\end{cases}
\end{equation*}
\begin{equation*}
\begin{cases}
\partial_t \rho^\varepsilon + ( v^\varepsilon , \nabla ) \rho^\varepsilon + ({\rm{div}} v^\varepsilon ) \rho^\varepsilon = 0 & \text{ in } \Omega^\varepsilon_T,\\
\rho^\varepsilon |_{t=0} = \rho_0 & \text{ in } \Omega (0).
\end{cases}
\end{equation*}
For each flow map $\tilde{x}^\varepsilon ( \xi , t )$, we set the \emph{action integral} as follows:
\begin{equation*}
A [ \tilde{x}^\varepsilon ] = - \int_0^T \int_{\Omega^\varepsilon (t)} \left\{ \frac{1}{2} \rho^\varepsilon | v^\varepsilon |^2 \right\} ( x , t) { \ } d x  d t .
\end{equation*}

\begin{theorem}[Variation of the flow map to Action Integral]\label{thm27}{ \ }\\
Suppose that $( \tilde{x}^\varepsilon (\xi ,t) , \Omega_T^\varepsilon)$ is a variation of $( \tilde{x} (\xi ,t ) , \Omega_T)$. Assume that for every $\xi \in \Gamma_0$ and $0 \leq t < T$,
\begin{equation*}
\rho^\varepsilon ( \tilde{x}^\varepsilon ( \xi , t) , t )|_{\varepsilon =0} = \rho ( \tilde{x}( \xi , t) , t ).
\end{equation*}
Then
\begin{equation*}
\frac{d}{d \varepsilon } \bigg|_{\varepsilon = 0} A [\tilde{x}^\varepsilon ] = \int_0^T \int_{\Omega ( t  )} \{ \rho D_t v \} (x,t) \cdot z (x,t) { \ } d x d t,
\end{equation*}
where $z = z (x,t)$ is a variation such that ${d}/{ d \varepsilon} |_{\varepsilon =0} \tilde{x}^\varepsilon ( \xi , t ) = \tilde{y} ( \xi , t)$ and \\$z ( \tilde{x} (\xi , t) , t) = \tilde{y} (\xi , t)$.
\end{theorem}

Applying Theorems \ref{thm21}-\ref{thm24} \ref{thm27}, Propositions \ref{prop25}, \ref{prop26}, an energetic variational approach, and the first law of thermodynamic, we can derive our non-Newtonian fluid systems. See Sections \ref{sect6} for details. See also Section \ref{sect6} for the enthalpy, entropy, free energy, and conservative form of the system \eqref{eq11}. 

Finally, we state conservation laws of our compressible fluid system \eqref{eq11}.
\begin{theorem}[Conservation laws]\label{thm28}{ \ }\\
Suppose that $\Omega (t)$ is flowed by the velocity $v$. Then the two assertions hold:\\
$(\mathrm{i})$ Assume that for each $0 < t < T$,
\begin{align*}
e'_4 (|{\rm{grad}} \theta |^2 ) (n , \nabla ) \theta |_{\partial \Omega (t)} & = 0,\\
e'_5 (|{\rm{grad}} C |^2 ) (n , \nabla ) C |_{\partial \Omega (t)} & = 0,
\end{align*}
and
\begin{multline*}
[e'_1 (| D_+ ( v) |^2) D_+(v) + e'_2 (| {\rm{div}} v |^2) ({\rm{div}} v ) I_{3 \times 3} \\
+ e'_3 (| D_- (v)|^2 ) D_-(v) - \sigma I_{3 \times 3} ] n |_{\partial \Omega (t)} = { }^t (0 , 0 , 0 ).
\end{multline*}
Then the system \eqref{eq11} satisfies that for $0 < t_1 < t_2 < T$,
\begin{align*}
\int_{\Omega (t_2)} \rho v { \ } d x & = \int_{\Omega (t_1)} \rho v { \ } d x + \int_{t_1}^{t_2} \int_{\Omega ( \tau )} \rho F { \ }d x d \tau ,\\
\int_{\Omega (t_2)} e_A { \ } d x & = \int_{\Omega (t_1)} e_A { \ } d x + \int_{t_1}^{t_2} \int_{ \Omega ( \tau )} \rho F \cdot v { \ }d x d \tau,\\
\int_{\Omega (t_2)} C { \ } d x & = \int_{\Omega (t_1)} C { \ } d x,
\end{align*}
where $e_A = \rho | v |^2 /2 + \rho e$. Here $n = n(x,t) = { }^t ( n_1 ,n_2 , n_3 )$ is the unit outer normal vector at $x \in \partial \Omega (t)$.\\
$(\mathrm{ii})$ Assume that $e_3 ( r ) = \mu_3 r$ for some $\mu_3 \in \mathbb{R}$. Suppose that
\begin{multline*}
[(e'_1 (| D_+ ( v) |^2) + \mu_3 )D_+(v) + (e'_2 (| {\rm{div}} v |^2) - \mu_3) ({\rm{div}} v ) I_{3} - \sigma I_{3} ] n |_{\partial \Omega (t)} = { }^t (0 , 0 , 0 ).
\end{multline*}
Then the system \eqref{eq11} satisfies that for $0 < t_1 < t_2 < T$,
\begin{equation*}
\int_{\Omega (t_2)} x \times \rho v { \ } d x = \int_{\Omega (t_1)} x \times \rho v { \ } d x + \int_{t_1}^{t_2} \int_{\Omega ( \tau )} x \times \rho F { \ }d x d \tau .
\end{equation*}
\end{theorem}

The outline of this paper as follows: In Section \ref{sect3}, we study the representation of our energy densities to prove Theorem \ref{thm21}. In Section \ref{sect4} we make use of a flow map and the Riemannian metric to prove Theorem \ref{thm27}. In Section \ref{sect5} we use the integration by parts to prove Theorems \ref{thm23} and \ref{thm24}. In Section \ref{sect6} we apply our energetic variational approaches to derive the generalized compressible and incompressible non-Newtonian fluid systems. Moreover, we study conservation laws of the compressible fluid system \eqref{eq11} to prove Theorem \ref{thm28}. In the Appendix $( \mathrm{I} )$ we give a method to derive strain rate tensors and fluxes from our energy densities. In the Appendix $(\mathrm{II})$ we state energetic variational approaches for the inviscid compressible and incompressible fluid systems.

\section{Representation of Energy Densities}\label{sect3}

In this section we discuss a mathematical validity of our energy densities for non-Newtonian fluid. To this end, we first introduce a flow map in a domain and the Riemannian metric induced by the flow map. Secondly, we derive fundamental properties of the flow map and Riemannian metric. Finally, we study the representation of our energy densities to prove Theorem \ref{thm21}.

\begin{definition}[Flow map in a domain]\label{def31}
Let $\Omega (t)$ be a domain in $\mathbb{R}^3$ depending on time $t \in [0,T)$ for some $T \in (0, \infty ]$, and $\tilde{x} = { }^t ( \tilde{x}_1, \tilde{x}_2 , \tilde{x}_3) \in [ C^3 ( \mathbb{R}^4) ]^3$. We call $\tilde{x} = \tilde{x} (\xi , t)$ a \emph{flow map} in $\Omega ( t )$ if the three properties hold:\\
$( \mathrm{i} )$ for every $\xi \in \Omega ( 0 )$
\begin{equation*}
\tilde{x} ( \xi , 0 ) = \xi,
\end{equation*}
$( \mathrm{ii})$ for all $\xi \in \Omega (0)$ and $0 \leq t < T$
\begin{equation*}
\tilde{x} ( \xi , t) \in \Omega ( t ),
\end{equation*}
$(\mathrm{iii})$ for each $0 \leq t < T$
\begin{equation*}
\tilde{x} ( \cdot , t): \Omega (0) \to \Omega (t) \text{ is bijective}.
\end{equation*}
\end{definition}

\begin{definition}[Velocity determined by a flow map]\label{def32}
Let $\Omega (t)$ be a domain in $\mathbb{R}^3$ depending on time $t \in [0,T)$ for some $T \in (0, \infty ]$, and let $\tilde{x} = \tilde{x} (\xi ,t)$ be a flow map in $\Omega ( t )$. Suppose that there is a smooth function $v = v (x , t) = { }^t (v_1 , v_2 , v_3 )$ such that for $\xi \in \Omega (0)$ and $0<t <T$,
\begin{equation*}
\frac{d \tilde{x}}{d t} = \tilde{x}_t ( \xi , t ) = v ( \tilde{x} ( \xi ,t ) ,t).
\end{equation*}
We call the vector-valued function $v$ the \emph{velocity} determined by the flow map $\tilde{x} (\xi , t)$.
\end{definition}

Let us now study fundamental properties of a flow map in a domain and the velocity determined by the flow map. Let $\Omega (t)$ be a bounded $C^2$-domain in $\mathbb{R}^3$ depending on time $t \in [0,T)$ for some $T \in (0, \infty ]$. Let $\tilde{x} = \tilde{x} (\xi ,t)$ be a flow map in $\Omega ( t )$, and let $v = v (x,t)$ be the velocity determined by the flow map $x$, i.e. for every $\xi \in \Omega (0)$ and $0<t <T$,
\begin{equation*}
\begin{cases}
\frac{d \tilde{x}}{d t} (\xi, t) = v (\tilde{x} ( \xi, t ) , t),\\
\tilde{x} (\xi, 0) = \xi .
\end{cases}
\end{equation*}
We assume that $v$ is the velocity of the fluid in the domain $\Omega (t)$. By the bijection of the flow map, we see that $\Omega (t)$ is expressed as follows:
\begin{equation*}
\Omega (t) = \{ x = { }^t (x_1, x_2 , x_3 ) \in \mathbb{R}^3;{ \ } x = \tilde{x} (\xi , t ) , {  \ } \xi \in \Omega (0) \}.
\end{equation*}
Using the change of variables, we see that for each smooth function $f = f (x,t)$,
\begin{equation*}
\int_{\Omega (t) } f (x , t) { \ }d x = \int_{\Omega (0)} f ( \tilde{x} ( \xi , t ) , t ) {\rm{det}} (\nabla_\xi \tilde{x}){ \ } d \xi.
\end{equation*}

Next we introduce the Riemannian metric induced by the flow map $\tilde{x} (\xi,t)$. For the flow map $\tilde{x} = \tilde{x} (\xi , t)$ in $\Omega (t)$,
\begin{equation*}
g_i =g_i (\xi,t) := \frac{\partial \tilde{x}}{\partial \xi_i} = { }^t \left(\frac{\partial \tilde{x}_1}{\partial \xi_i},\frac{\partial \tilde{x}_2}{\partial \xi_i},\frac{\partial \tilde{x}_3}{\partial \xi_i} \right).
\end{equation*}
Write
\begin{equation*}
g_{i j}=g_{i j} (\xi,t) := g_i \cdot g_j = \frac{\partial \tilde{x}_\ell }{\partial \xi_i} \frac{\partial \tilde{x}_\ell}{\partial \xi_j}= \sum_{\ell = 1}^3 \frac{\partial \tilde{x}_\ell}{\partial \xi_i} \frac{\partial \tilde{x}_\ell}{\partial \xi_j}.
\end{equation*}
Set
\begin{align*}
&( g^{i j} )_{3 \times 3} := ((g_{i j})_{3 \times 3} )^{-1}, \text{that is, } \begin{pmatrix}
g^{11} & g^{12} & g^{13}\\
g^{21} & g^{22} & g^{23}\\
g^{31} & g^{32} & g^{33}
\end{pmatrix} := \begin{pmatrix}
g_{11} & g_{12} & g_{13}\\
g_{21} & g_{22} & g_{23}\\
g_{31} & g_{32} & g_{33}
\end{pmatrix}^{-1},\\
&g^i := g^{i j}g_j = g^{i 1}g_1 + g^{i 2}g_2 + g^{i 3} g_3,\\
&\acute{g}_i := \frac{d}{d t} g_i  = \frac{\partial v}{\partial \xi_i} = { }^t \left( \frac{\partial v_1}{\partial \xi_i} , \frac{\partial v_2}{\partial \xi_i} ,  \frac{\partial v_3}{\partial \xi_i} \right).
\end{align*}
It is easy to check that $g_{j i} = g_{ i j}$, $g^{j i} = g^{i j}$, $g^{i j} = g^i \cdot g^j$, $g^i \cdot g_j = \delta_{i j}$,
\begin{align*}
&g_i = g_{i j}g^j = g_{i 1}g^1 + g_{i 2}g^2 +g_{i 3}g^3,\\
&\acute{g}_{i j} = \acute{g}_i \cdot g_j + g_i \cdot \acute{g}_j,\\
&\acute{g}_i = \frac{\partial v}{\partial \xi_i} = \frac{\partial \tilde{x}_\ell}{\partial \xi_i}\frac{\partial v}{\partial \tilde{x}_\ell},\\
& \acute{g}_i \cdot g_j  = \frac{\partial \tilde{x}_k}{\partial \xi_i}\frac{\partial v_\ell}{\partial \tilde{x}_k}\frac{\partial \tilde{x}_\ell}{\partial \xi_j},\\
& \acute{g}_{i j} = 2 \frac{\partial \tilde{x}_k}{\partial \xi_i} [D_+ (v)]_{k \ell} \frac{\partial \tilde{x}_\ell}{\partial \xi_j} ,
\end{align*}
where $\delta_{i j}$ is Kronecker's delta and $D_+ (v) = \{ (\nabla v) + { }^t (\nabla v) \}/2$. Indeed, we see at once that
\begin{equation*}
g^i \cdot g_j = (g^{i 1} g_1 + g^{i 2} g_2 + g^{i 3} g_3 ) \cdot g_j = g^{i 1 } g_{1 j} + g^{i 2} g_{2 j} + g^{i 3} g_{3 j} = \delta_{i j}
\end{equation*}
and that
\begin{align*}
\acute{g}_i \cdot g_j & = { }^t \left( \sum_{i=1}^3 \frac{\partial \tilde{x}_k}{\partial \xi_i}\frac{\partial v_1}{\partial \tilde{x}_k} , \sum_{i=1}^3 \frac{\partial \tilde{x}_k}{\partial \xi_i}\frac{\partial v_2}{\partial \tilde{x}_k}  ,  \sum_{i=1}^3 \frac{\partial \tilde{x}_k}{\partial \xi_i}\frac{\partial v_3}{\partial \tilde{x}_k}  \right) \cdot { }^t \left(\frac{\partial \tilde{x}_1}{\partial \xi_j},\frac{\partial \tilde{x}_2}{\partial \xi_j},\frac{\partial \tilde{x}_3}{\partial \xi_j} \right)\\
& = \sum_{i , j = 1}^3 \frac{\partial \tilde{x}_k}{\partial \xi_i}\frac{\partial v_\ell}{\partial \tilde{x}_k}\frac{\partial \tilde{x}_\ell}{\partial \xi_j} = \frac{\partial \tilde{x}_k}{\partial \xi_i}\frac{\partial v_\ell}{\partial \tilde{x}_k}\frac{\partial \tilde{x}_\ell}{\partial \xi_j}.
\end{align*}
Since ${ }^t (\nabla_\xi \tilde{x} ) (\nabla_\xi \tilde{x}) = (g_{i j})_{3 \times 3}$ and ${\rm{det}} ({ }^t (\nabla_\xi \tilde{x} )) = {\rm{det}} (\nabla_\xi \tilde{x})$, we find that
\begin{equation*}
{\rm{det}} (\nabla_\xi \tilde{x} ) = \sqrt{ {\rm{det}} (g_{i j})_{3 \times 3} }.
\end{equation*}
From now on we set
\begin{equation*}
J = J ( \xi , t)  = \sqrt{{\rm{det}} (g_{i j})_{3 \times 3}},
\end{equation*}
and we assume that $J >0$.

Next we study basic properties of flow maps.
\begin{lemma}[Properties of Riemannian metric induced by flow map]\label{lem33} { \ }\\
Let $f \in C^1 ( \mathbb{R}^4)$. Then
\begin{align}
\delta_{i j } & = \frac{\partial \tilde{x}_i }{\partial \xi_k} \frac{\partial \tilde{x}_j }{\partial \xi_\ell} g^{k \ell} \label{eq31},\\
\int_{\Omega (t)} \frac{\partial f}{\partial x_i} { \ }d x & = \int_{\Omega (0)} g^{k \ell} \frac{\partial \tilde{x}_i}{\partial \xi_k} \frac{\partial f}{\partial \xi_\ell} J { \ } d \xi,\label{eq32}\\
\int_{\Omega (t)} | \nabla f |^2 { \ }d x & = \int_{\Omega (0)} g^{k \ell} \frac{\partial f}{\partial \xi_k} \frac{\partial f}{\partial \xi_\ell} J { \ } d \xi , \label{eq33}\\
\frac{d }{d t } \sqrt{ {\rm{det}} (g_{i j})_{ 3 \times 3} }& = ( \acute{g}_k \cdot g^k ) \sqrt{{\rm{det}} (g_{i j})_{3 \times 3} },\label{eq34}\\
\int_{\Omega (t)} f ({\rm{div}} v ) { \ } d x & = \int_{\Omega (0)} f ( \acute{g}_j \cdot g^j) J { \ } d \xi,\label{eq35}\\
\int_{\Omega (t)} f ({\rm{div}} v ) { \ } d x & = \int_{\Omega (0)} f \frac{d J}{d t} { \ } d \xi .\label{eq36}
\end{align}
\end{lemma}

\begin{proof}[Proof of Lemma \ref{lem33}]

We first prove \eqref{eq31}, \eqref{eq32}, and \eqref{eq33}. Set
\begin{equation*}
M = { }^t ( \nabla_\xi \tilde{x}) =
\begin{pmatrix}
\frac{\partial \tilde{x}_1 }{\partial \xi_1} & \frac{\partial \tilde{x}_2 }{\partial \xi_1} & \frac{\partial \tilde{x}_3 }{\partial \xi_1}\\
\frac{\partial \tilde{x}_1 }{\partial \xi_2} & \frac{\partial \tilde{x}_2 }{\partial \xi_2} & \frac{\partial \tilde{x}_3 }{\partial \xi_2}\\
\frac{\partial \tilde{x}_1 }{\partial \xi_3} & \frac{\partial \tilde{x}_2 }{\partial \xi_3} & \frac{\partial \tilde{x}_3 }{\partial \xi_3}
\end{pmatrix}.
\end{equation*}
Since $M { }^t M = (g_{i j})_{3 \times 3}$, it follows from the definition of $g^{i j}$ to have
\begin{equation*}
M { }^t M (g^{i j })_{3 \times 3} = I_{3 \times 3}.
\end{equation*}
We left-multiply the both sides of the above equality by $M^{-1}$ to see that
\begin{equation*}
{ }^t M (g^{i j })_{3 \times 3} = M^{-1}.
\end{equation*}
This gives
\begin{equation*}
{ }^t M (g^{i j })_{3 \times 3} M = I_{3 \times 3}.
\end{equation*}
Therefore we find that
\begin{equation*}
\delta_{i j} = [I_{3 \times 3}]_{i j} = [{ }^t M (g^{i j })_{3 \times 3} M]_{i j} = \frac{\partial \tilde{x}_i }{\partial \xi_k} \frac{\partial \tilde{x}_j }{\partial \xi_\ell} g^{k \ell},
\end{equation*}
which is \eqref{eq31}. Set $Q = M^{-1}$. Then $({ }^t M)^{-1} = { }^t Q$. Since $(g^{i j})_{3 \times 3} = ((g_{i j})_{3 \times 3})^{-1}$, we see that
\begin{equation*}
{ }^t Q Q = (g^{i j})_{3 \times 3}.
\end{equation*}
Therefore we find that
\begin{align*}
Q = { }^t M { }^t Q Q & = { }^t M (g^{i j})_{3 \times 3}\\
& =
\begin{pmatrix}
\frac{\partial \tilde{x}_1}{\partial \xi_k} g^{1 k} & \frac{\partial \tilde{x}_1}{\partial \xi_k} g^{2 k}  & \frac{\partial \tilde{x}_1}{\partial \xi_k} g^{3 k} \\
\frac{\partial \tilde{x}_2}{\partial \xi_k} g^{1 k} & \frac{\partial \tilde{x}_2}{\partial \xi_k} g^{2 k}  & \frac{\partial \tilde{x}_2}{\partial \xi_k} g^{3 k} \\
\frac{\partial \tilde{x}_3}{\partial \xi_k} g^{1 k} & \frac{\partial \tilde{x}_3}{\partial \xi_k} g^{2 k}  & \frac{\partial \tilde{x}_3}{\partial \xi_k} g^{3 k}
\end{pmatrix}.
\end{align*}
Since
\begin{equation*}
\begin{pmatrix}
\frac{\partial f}{\partial \xi_1}\\
\frac{\partial f}{\partial \xi_2}\\
\frac{\partial f}{\partial \xi_3}
\end{pmatrix}
=
\begin{pmatrix}
\frac{\partial \tilde{x}_1 }{\partial \xi_1} & \frac{\partial \tilde{x}_2 }{\partial \xi_1} & \frac{\partial \tilde{x}_3 }{\partial \xi_1}\\
\frac{\partial \tilde{x}_1 }{\partial \xi_2} & \frac{\partial \tilde{x}_2 }{\partial \xi_2} & \frac{\partial \tilde{x}_3 }{\partial \xi_2}\\
\frac{\partial \tilde{x}_1 }{\partial \xi_3} & \frac{\partial \tilde{x}_2 }{\partial \xi_3} & \frac{\partial \tilde{x}_3 }{\partial \xi_3}
\end{pmatrix}
\begin{pmatrix}
\frac{\partial f}{\partial \tilde{x}_1}\\
\frac{\partial f}{\partial \tilde{x}_2}\\
\frac{\partial f}{\partial \tilde{x}_3}
\end{pmatrix}= M
\begin{pmatrix}
\frac{\partial f}{\partial \tilde{x}_1}\\
\frac{\partial f}{\partial \tilde{x}_2}\\
\frac{\partial f}{\partial \tilde{x}_3}
\end{pmatrix},
\end{equation*}
we left-multiply both sides of the above equality by $Q (= M^{-1})$ to see that
\begin{align*}
\begin{pmatrix}
\frac{\partial f}{\partial \tilde{x}_1}\\
\frac{\partial f}{\partial \tilde{x}_2}\\
\frac{\partial f}{\partial \tilde{x}_3}
\end{pmatrix} = Q M \begin{pmatrix}
\frac{\partial f}{\partial \tilde{x}_1}\\
\frac{\partial f}{\partial \tilde{x}_2}\\
\frac{\partial f}{\partial \tilde{x}_3}
\end{pmatrix} &= Q\begin{pmatrix}
\frac{\partial f}{\partial \xi_1}\\
\frac{\partial f}{\partial \xi_2}\\
\frac{\partial f}{\partial \xi_3}
\end{pmatrix} \\
& =
\begin{pmatrix}
\frac{\partial \tilde{x}_1}{\partial \xi_k} g^{1 k} & \frac{\partial \tilde{x}_1}{\partial \xi_k} g^{2 k}  & \frac{\partial \tilde{x}_1}{\partial \xi_k} g^{3 k} \\
\frac{\partial \tilde{x}_2}{\partial \xi_k} g^{1 k} & \frac{\partial \tilde{x}_2}{\partial \xi_k} g^{2 k}  & \frac{\partial \tilde{x}_2}{\partial \xi_k} g^{3 k} \\
\frac{\partial \tilde{x}_3}{\partial \xi_k} g^{1 k} & \frac{\partial \tilde{x}_3}{\partial \xi_k} g^{2 k}  & \frac{\partial \tilde{x}_3}{\partial \xi_k} g^{3 k}
\end{pmatrix}
\begin{pmatrix}
\frac{\partial f}{\partial \xi_1}\\
\frac{\partial f}{\partial \xi_2}\\
\frac{\partial f}{\partial \xi_3}
\end{pmatrix}
=
\begin{pmatrix}
\frac{\partial \tilde{x}_1}{\partial \xi_k} g^{\ell k} \frac{\partial f}{\partial \xi_\ell}\\
\frac{\partial \tilde{x}_2}{\partial \xi_k} g^{\ell k} \frac{\partial f}{\partial \xi_\ell}\\
\frac{\partial \tilde{x}_3}{\partial \xi_k} g^{\ell k} \frac{\partial f}{\partial \xi_\ell}
\end{pmatrix}.
\end{align*}
Therefore we have \eqref{eq32}. Using \eqref{eq32}, we observe that
\begin{align*}
| \nabla f ( \tilde{x}( \xi , t) , t ) |^2 & = \begin{pmatrix}
\frac{\partial \tilde{x}_1}{\partial \xi_k} g^{\ell k} \frac{\partial f}{\partial \xi_\ell}\\
\frac{\partial \tilde{x}_2}{\partial \xi_k} g^{\ell k} \frac{\partial f}{\partial \xi_\ell}\\
\frac{\partial \tilde{x}_3}{\partial \xi_k} g^{\ell k} \frac{\partial f}{\partial \xi_\ell}
\end{pmatrix}
\cdot
\begin{pmatrix}
\frac{\partial \tilde{x}_1}{\partial \xi_{k'}} g^{\ell' k'} \frac{\partial f}{\partial \xi_{\ell'}}\\
\frac{\partial \tilde{x}_2}{\partial \xi_{k'}} g^{\ell' k'} \frac{\partial f}{\partial \xi_{\ell'}}\\
\frac{\partial \tilde{x}_3}{\partial \xi_{k'}} g^{\ell' k'} \frac{\partial f}{\partial \xi_{\ell'}}
\end{pmatrix}\\
& = g_{k k'} g^{\ell k} g^{\ell' k'} \frac{\partial f}{\partial \xi_{\ell}}\frac{\partial f}{\partial \xi_{\ell'}}\\
& = g^{\ell \ell'} \frac{\partial f}{\partial \xi_{\ell}}\frac{\partial f}{\partial \xi_{\ell'}}.
\end{align*}
Therefore we see \eqref{eq33}.

Next we show \eqref{eq34}, \eqref{eq35}, and \eqref{eq36}. By the definition of $g^i $, we see that
\begin{align*}
&{\rm{Tr}} \left( (g^{i j})_{3 \times 3} \frac{d (g_{i j})_{3 \times 3}}{d t} \right)\\
& = {\rm{Tr}} \begin{pmatrix}
g^{1 k} \acute{g}_k \cdot g_1 + g^{1 k} g_k \cdot \acute{g}_1 & g^{1 k} \acute{g}_k \cdot g_2 + g^{1 k} g_k \cdot \acute{g}_2 & g^{1 k} \acute{g}_k \cdot g_3 + g^{1 k} g_k \cdot \acute{g}_3\\
g^{2 k} \acute{g}_k \cdot g_1 + g^{2 k} g_k \cdot \acute{g}_1 & g^{2 k} \acute{g}_k \cdot g_2 + g^{2 k} g_k \cdot \acute{g}_2 & g^{2 k} \acute{g}_k \cdot g_3 + g^{2 k} g_k \cdot \acute{g}_3\\
g^{3 k} \acute{g}_k \cdot g_1 + g^{3 k} g_k \cdot \acute{g}_1 &  g^{3 k} \acute{g}_k \cdot g_2 + g^{3 k} g_k \cdot \acute{g}_2 & g^{3 k} \acute{g}_k \cdot g_3 + g^{3 k} g_k \cdot \acute{g}_3
\end{pmatrix}\\
& = g^{\ell k} \acute{g}_k \cdot g_\ell + g^{\ell k} g_k \cdot \acute{g}_\ell = 2 \acute{g}_k \cdot g^k.
\end{align*}
Since
\begin{equation*}
\frac{d}{d t} {\rm{det}}(g_{i j})_{3 \times 3} = {\rm{det}} (g_{ i j})_{3 \times 3} {\rm{Tr}} \left( (g^{i j})_{3 \times 3} \frac{d (g_{i j})_{3 \times 3}}{d t} \right),
\end{equation*}
we check that
\begin{align*}
\frac{d}{d t} \sqrt{{\rm{det}}(g_{i j})_{3 \times 3} } & = \frac{1}{2 \sqrt{{\rm{det}}(g_{i j})_{3 \times 3}} } \frac{d}{d t} {\rm{det}} (g_{i j})_{3 \times 3} \\
& =  \frac{1}{2 \sqrt{{\rm{det}}(g_{i j})_{3 \times 3}} } {\rm{det}} (g_{i j})_{3 \times 3} (2 \acute{g}_k \cdot g^k)\\
& = \sqrt{{\rm{det}} (g_{i j})_{3 \times 3} } (\acute{g}_k \cdot g^k) .
\end{align*}
Therefore we have \eqref{eq34}. Here we used the fact that
\begin{equation*}
\frac{d}{d t} {\rm{det}}A = ({\rm{det}}A ){\rm{Tr}} \left( A^{-1} \frac{d A}{d t} \right), \text{ where }A = A (t) \text{ is a regular matrix.} 
\end{equation*}
A direct calculation shows that
\begin{align*}
\left\{ \frac{\partial v_1}{\partial x_1} + \frac{\partial v_2}{\partial x_2} + \frac{\partial v_3}{\partial x_3} \right\} (\tilde{x} ( \xi ,t ) , t)& = g^{k \ell} \left( \frac{\partial \tilde{x}_1}{\partial \xi_k} \frac{\partial v_1}{\partial \xi_\ell} +  \frac{\partial \tilde{x}_2}{\partial \xi_k} \frac{\partial v_2}{\partial \xi_\ell}  +  \frac{\partial \tilde{x}_3}{\partial \xi_k} \frac{\partial v_3}{\partial \xi_\ell}  \right)\\
&= g^{k \ell } g_k \cdot \acute{g}_\ell\\
& = g^\ell \cdot \acute{g}_\ell .
\end{align*}
This implies that
\begin{equation*}
\int_{\Omega (t)} f ({\rm{div}} v ) { \ }d x  = \int_{\Omega (0)} f (g^\ell \cdot \acute{g}_\ell ) J { \ } d \xi  .
\end{equation*}
Therefore we have \eqref{eq35}. Combining \eqref{eq34} and \eqref{eq35}, we see \eqref{eq36}. Therefore the lemma follows.

\end{proof}

To prove Theorem \ref{thm21}, we prepare the following lemma.
\begin{lemma}\label{lem34}
\begin{align}
\int_{\Omega (t)} ({\rm{div}} v ) { \ }d x & = \int_{\Omega (0)} \frac{1}{2} \acute{g}_{i j} g^{i j} J { \ }d \xi,\label{eq37}\\
\int_{\Omega (t)} | {\rm{div}} v |^2 { \ }d x & = \int_{\Omega (0)} \frac{1}{4} \acute{g}_{i j} \acute{g}_{k \ell} g^{i j} g^{k \ell} J { \ }d \xi, \label{eq38}\\
\int_{\Omega (t)} |D_+ (v)|^2 { \ }d x & = \int_{\Omega (0)} \frac{1}{4} \acute{g}_{i j}  \acute{g}_{k \ell} g^{i k} g^{j \ell} J { \ }d \xi,\label{eq39}\\
\int_{\Omega (t)} | D_- (v) |^2 { \ }d x & = \int_{\Omega (0)} \left\{ (\acute{g}_i \cdot \acute{g}_j ) ( g^i \cdot g^j) - \frac{1}{4} \acute{g}_{i j} \acute{g}_{k \ell} g^{i k} g^{j \ell} ) \right\} J{ \ }d \xi. \label{eq310}
\end{align}
\end{lemma}

\begin{proof}[Proof of Lemma \ref{lem34}]

We first prove \eqref{eq37}. Since
\begin{align*}
\acute{g}_{i j} g^{i j} & = (\acute{g}_i \cdot g_j + g_i \cdot \acute{g}_j ) (g^i \cdot g^j)\\
& = (\acute{g}_i \cdot g^i ) + ( \acute{g}_j \cdot g^j) = 2 \acute{g}_i \cdot g^i,
\end{align*}
it follows from \eqref{eq35} to have \eqref{eq37}. Using \eqref{eq37}, we check that
\begin{align*}
\acute{g}_{i j} \acute{g}_{k \ell} g^{i j} g^{k \ell } & = 4 (\acute{g}_i \cdot g^i )( \acute{g}_k \cdot g^k) = 4 |{\rm{div}} v |^2 .
\end{align*}
This is \eqref{eq38}. An easy calculation gives
\begin{align*}
\frac{1}{4} \acute{g}_{i j} \acute{g}_{k \ell} g^{i k} g^{j \ell} = & \left( \frac{\partial \tilde{x}_{i'} }{\partial \xi_{i} } [D_+ (v) ]_{i' j'}\frac{\partial \tilde{x}_{j'} }{\partial \xi_{j} } \right) \left( \frac{\partial \tilde{x}_{k'} }{\partial \xi_{k} } [D_+ (v) ]_{k' \ell'}\frac{\partial \tilde{x}_{\ell'} }{\partial \xi_{\ell} } \right)  g^{i k} g^{j \ell}\\
= & \left( \frac{\partial \tilde{x}_{i'} }{\partial \xi_{i} } \frac{\partial \tilde{x}_{k'} }{\partial \xi_{k} } g^{i k} \right) \left( \frac{\partial \tilde{x}_{j'} }{\partial \xi_{j} }  \frac{\partial \tilde{x}_{\ell'} }{\partial \xi_{\ell} } g^{j \ell} \right)  [D_+ (v) ]_{i' j'} [D_+ (v) ]_{k' \ell'}\\
= & \delta_{i' k'} \delta_{j' \ell'}  [D_+ (v) ]_{i' j'} [D_+ (v) ]_{k' \ell'} = D_+ (v): D_+ (v).
\end{align*}
Therefore we have \eqref{eq39}. Finally we prove \eqref{eq310}. Using \eqref{eq33}, we see that
\begin{align*}
| \nabla v |^2 = \{ | \nabla v_1 |^2 + | \nabla v_2 |^2 + | \nabla v_3 |^2 \} ( \tilde{x}( \xi , t ) , t )  & = g^{k \ell} \frac{\partial v_j}{\partial \xi_k} \frac{\partial v_j}{\partial \xi_\ell}\\
& = g^{k \ell} \acute{g}_k \cdot \acute{g}_\ell\\
& = (\acute{g}_k \cdot \acute{g}_{\ell} )( g^k \cdot g^\ell ) .
\end{align*}
Since
\begin{equation*}
|D_- (v)^2| = |\nabla v|^2 - | D_+ ( v )|^2,
\end{equation*}
we see that
\begin{equation*}
| D_- (v) |^2 =  (\acute{g}_k \cdot \acute{g}_{\ell} )( g^k \cdot g^\ell ) - \frac{1}{4} \acute{g}_{i j} \acute{g}_{k \ell} g^{i k} g^{ j \ell }.
\end{equation*}
Therefore the lemma follows.

\end{proof}

\begin{proof}[Proof of Theorem \ref{thm21}]
Applying Lemmas \ref{lem33} and \ref{lem34}, we prove Theorem \ref{thm21}.
\end{proof}

\section{Variation of the Flow Map to Action Integral}\label{sect4}

We study a variation of the action integral with respect to the flow maps. Let $\Omega (t)$ be a bounded $C^2$-domain in $\mathbb{R}^3$ depending on time $t \in [0,T)$ for some $T \in (0, \infty ]$. We introduce a flow map in a variation of the domain $\Omega (t)$.

\begin{definition}[Flow map in a variation of the domain $\Omega (t)$]\label{def41}{ \ }\\
For $- 1 < \varepsilon <1$, let $\Omega^\varepsilon (t)$ be a variation of $\Omega (t)$. Let $\tilde{x}^\varepsilon = { }^t ( \tilde{x}^\varepsilon_1 , \tilde{x}^\varepsilon_2 , \tilde{x}^\varepsilon_3 ) \in [ C^3 ( \mathbb{R}^4) ]^3$. We call $\tilde{x}^\varepsilon = \tilde{x}^\varepsilon ( \xi , t )$ a \emph{flow map} in $\Omega^\varepsilon ( t )$ if the three properties hold:\\
$( \mathrm{i} )$ for every $\xi \in \Omega (0) (= \Omega ( 0 ))$
\begin{equation*}
\tilde{x}^\varepsilon ( \xi , 0 ) = \xi,
\end{equation*}
$( \mathrm{ii})$ for all $\xi \in \Omega (0)$ and $0 \leq t < T$
\begin{equation*}
\tilde{x}^\varepsilon ( \xi , t) \in \Omega^\varepsilon ( t ),
\end{equation*}
$(\mathrm{iii})$ for each $0 \leq t < T$
\begin{equation*}
\tilde{x}^\varepsilon ( \cdot , t): \Omega (0) \to \Omega^\varepsilon (t) \text{ is bijective}.
\end{equation*}
\end{definition}
Note that from the property $(\mathrm{iii})$ we can write
\begin{equation*}
\Omega^\varepsilon (t) = \{  x \in \mathbb{R}^3;{ \ }x = \tilde{x}^\varepsilon (\xi , t ) , { \ }\xi \in \Omega (0) \} .
\end{equation*}

\begin{definition}[Velocity determined by a flow map in $\Omega^\varepsilon (t)$]\label{def42}{ \ }\\
For $- 1 < \varepsilon <1$, let $\Omega^\varepsilon (t)$ be a variation of $\Omega (t)$. Let $\tilde{x}^\varepsilon = \tilde{x}^\varepsilon (\xi ,t)$ be a flow map in $\Omega^\varepsilon ( t )$. Suppose that there is a smooth function $v^\varepsilon = v^\varepsilon (x , t) = { }^t (v_1^\varepsilon , v_2^\varepsilon , v_3^\varepsilon )$ such that for $\xi \in \Omega (0)$ and $0<t <T$,
\begin{equation*}
\frac{d \tilde{x}^\varepsilon}{d t} = \tilde{x}^\varepsilon_t ( \xi , t ) = v^\varepsilon ( \tilde{x}^\varepsilon ( \xi ,t ) ,t) .
\end{equation*}
We call the vector-valued function $v^\varepsilon$ the \emph{velocity} determined by the flow map $\tilde{x}^\varepsilon (\xi ,t)$.
\end{definition}

For $- 1< \varepsilon <1 $, let $\Omega^\varepsilon (t)$ be a variation of $\Omega ( t )$. Let $\tilde{x}^\varepsilon = \tilde{x}^\varepsilon (\xi ,t)$ be a flow map in $\Omega^\varepsilon ( t )$, and let $v^\varepsilon = v^\varepsilon (x,t)$ be the velocity determined by the flow map $\tilde{x}^\varepsilon$, i.e. for $\xi \in \Omega (0)$ and $0<t <T$,
\begin{equation*}
\begin{cases}
\frac{d \tilde{x}^\varepsilon}{d t} (\xi, t) = v^\varepsilon (\tilde{x}^\varepsilon ( \xi, t ) , t),\\
\tilde{x}^\varepsilon (\xi, 0) = \xi .
\end{cases}
\end{equation*}

For each flow map $\tilde{x}^\varepsilon = \tilde{x}^\varepsilon (\xi , t)$ in $\Omega^\varepsilon (t)$,
\begin{equation*}
g_i^\varepsilon := \frac{\partial g}{\partial \xi_i} = { }^t \left(\frac{\partial \tilde{x}^\varepsilon_1}{\partial \xi_i},\frac{\partial \tilde{x}^\varepsilon_2}{\partial \xi_i},\frac{\partial \tilde{x}^\varepsilon_3}{\partial \xi_i} \right).
\end{equation*}
Write
\begin{equation*}
g^\varepsilon_{i j}:= g_i^\varepsilon \cdot g_j^\varepsilon = \frac{\partial \tilde{x}^\varepsilon_\ell}{\partial \xi_i} \frac{\partial \tilde{x}^\varepsilon_\ell}{\partial \xi_j}= \sum_{\ell =1}^3 \frac{\partial \tilde{x}^\varepsilon_\ell}{\partial \xi_i} \frac{\partial \tilde{x}^\varepsilon_\ell}{\partial \xi_j}.
\end{equation*}
Set
\begin{align*}
&( g^{i j}_\varepsilon )_{3 \times 3} := ((g_{i j}^\varepsilon)_{3 \times 3} )^{-1}, \text{that is, } \begin{pmatrix}
g^{11}_\varepsilon & g^{12}_\varepsilon & g^{13}_\varepsilon \\
g^{21}_\varepsilon & g^{22}_\varepsilon & g^{23}_\varepsilon \\
g^{31}_\varepsilon & g^{32}_\varepsilon & g^{33}_\varepsilon 
\end{pmatrix} := \begin{pmatrix}
g_{11}^\varepsilon & g_{12}^\varepsilon & g_{13}^\varepsilon \\
g_{21}^\varepsilon & g_{22}^\varepsilon & g_{23}^\varepsilon \\
g_{31}^\varepsilon & g_{32}^\varepsilon & g_{33}^\varepsilon 
\end{pmatrix}^{-1},\\
&g^i_\varepsilon := g_\varepsilon^{i j}g_j^\varepsilon = g_\varepsilon^{i 1}g_1^\varepsilon + g_\varepsilon^{i 2}g_2^\varepsilon + g_\varepsilon^{i 3}g_3^\varepsilon,\\
&\acute{g}^\varepsilon_i := \frac{d}{d t} g^\varepsilon_i  = \frac{\partial v^\varepsilon}{\partial \xi_i} = { }^t \left( \frac{\partial v^\varepsilon_1}{\partial \xi_i} , \frac{\partial v^\varepsilon_2}{\partial \xi_i} ,  \frac{\partial v^\varepsilon_3}{\partial \xi_i} \right).
\end{align*}
Changing the variables, we see that for $f = f (x,t)$
\begin{equation*}
\int_{\Omega^\varepsilon (t)} f ( x , t) { \ }d x = \int_{\Omega (0)} f ( \tilde{x}^\varepsilon (\xi , t ), t) \sqrt{{\rm{det}} (g_{i j}^\varepsilon )_{3 \times 3} } { \ }d \xi .
\end{equation*}
Set
\begin{equation*}
J^\varepsilon = J^\varepsilon ( \xi , t) = \sqrt{ {\rm{det}} (g^{\varepsilon}_{i j})_{3 \times 3}}.
\end{equation*}
We assume that $J^\varepsilon >0$. By the argument similar to derive \eqref{eq36}, we see that
\begin{equation}\label{eq41}
\int_{\Omega^\varepsilon (t)} f ({\rm{div}} v^\varepsilon ) { \ }d x = \int_{\Omega (0)} f \frac{d J^\varepsilon}{d t} { \ } d \xi .
\end{equation}

Let us first prove Proposition \ref{prop26}.

\begin{proof}[Proof of Proposition \ref{prop26}]
We only prove $(\mathrm{ii})$. Fix $U^\varepsilon (t) \subset \Omega^\varepsilon (t)$. By the bijective of the flow map $\tilde{x}^\varepsilon = \tilde{x}^\varepsilon ( \xi , t )$, there is $U_0 \subset \Omega (0)$ such that
\begin{equation*}
U^\varepsilon (t) = \{ x \in \mathbb{R}^3; { \ } x = x^\varepsilon (\xi , t), { \ }\xi \in U_0 \}.
\end{equation*} 
Since
\begin{align*}
\int_{U^\varepsilon (t)} \rho^\varepsilon ( x , t) { \ }d x = \int_{U_0} \rho^\varepsilon ( \tilde{x}^\varepsilon (\xi, t ) , t) J^\varepsilon  { \ } d \xi,
\end{align*}
we use \eqref{eq41} to see that
\begin{align*}
\frac{d }{d t} \int_{U^\varepsilon (t)} \rho^\varepsilon ( x , t) { \ }d x & = \int_{U_0} \left( \frac{d}{d t} \{ \rho^\varepsilon ( \tilde{x}^\varepsilon (\xi, t ) , t) \} J^\varepsilon +  \rho^\varepsilon ( \tilde{x} (\xi, t ) , t) \frac{d J^\varepsilon }{d t}   \right) { \ } d \xi\\
& = \int_{U^\varepsilon (t)} \{ \partial_t \rho^\varepsilon + (v^\varepsilon , \nabla ) \rho^\varepsilon + ({\rm{div}} v^\varepsilon ) \rho^\varepsilon \} { \ } d x .
\end{align*}
Since we can choose $U^\varepsilon (t) \subset \Omega^\varepsilon (t)$ arbitrarily, we conclude that
\begin{equation*}
\partial_t \rho^\varepsilon + (v^\varepsilon , \nabla ) \rho^\varepsilon + ({\rm{div}} v^\varepsilon ) \rho^\varepsilon = 0.
\end{equation*}
Therefore Proposition \ref{prop26} is proved.
\end{proof}

\begin{lemma}[Representation of energies]\label{lem43}{ \ }\\
Let $\rho_0 = \rho_0 ( x )$ be a smooth function. Then the following two assertions hold:\\
$(\mathrm{i})$ Assume that 
\begin{equation*}
\begin{cases}
\partial_t \rho + ( v , \nabla ) \rho + ({\rm{div}} v ) \rho = 0 & \text{ in } \Omega_T,\\
\rho |_{t = 0} = \rho_0 & \text{ in }\Omega_0.
\end{cases}
\end{equation*}
Then
\begin{align}
\int_{\Omega (t)} \frac{1}{2} \rho | v |^2 { \ }d x & = \int_{\Omega (0)} \frac{1}{2} \rho_0 ( \xi ) \tilde{x}_t \cdot \tilde{x}_t { \ } d \xi,\label{eq42}\\
\int_{\Omega (t)} \rho e { \ }d x & = \int_{\Omega (0)} \rho_0 ( \xi ) e { \ } d \xi,\label{eq43}\\
\int_{\Omega (t) } \rho F \cdot v { \ } d x & = \int_{\Omega (0)} \rho_0 (\xi ) F \cdot \tilde{x}_t { \ } d \xi \label{eq44}.
\end{align}
$(\mathrm{ii})$ Assume that
\begin{equation*}
\begin{cases}
\rho_t^\varepsilon + ( v^\varepsilon , \nabla ) \rho^\varepsilon + ({\rm{div}} v^\varepsilon ) \rho^\varepsilon = 0 & \text{ in } \Omega_T^\varepsilon,\\
\rho^\varepsilon|_{t = 0} = \rho_0 & \text{ in } \Omega_0.
\end{cases}
\end{equation*}
Then
\begin{equation}
\int_{ \Omega^\varepsilon (t) } \frac{1}{2} \rho^\varepsilon | v^\varepsilon |^2 { \ } d x = \int_{\Omega (0)} \frac{1}{2} \rho_0 ( \xi ) \tilde{x}_t^\varepsilon \cdot \tilde{x}_t^\varepsilon { \ }d \xi. \label{eq45}
\end{equation}

\end{lemma}

\begin{proof}[Proof of Lemma \ref{lem43}]
From \eqref{eq36} and \eqref{eq35}, we check that
\begin{equation*}
\frac{d }{d t} \{ \rho ( \tilde{x} ( \xi , t ) , t ) J ( \xi , t) \} = \{ \partial_t \rho + ( v , \nabla ) \rho  + ({\rm{div}} v ) \rho \} J = 0. 
\end{equation*}
Integrating with respect to time, we see that
\begin{equation*}
\rho (\tilde{x} ( \xi , t) ,t ) J (\xi , t) = \rho_0 (\xi ).
\end{equation*}
Thus, we have
\begin{equation*}
\rho ( \tilde{x} (\xi , t) , t ) = \frac{\rho_0 ( \xi )}{J (\xi , t )}.
\end{equation*}
Similarly, we see that
\begin{equation*}
\rho^\varepsilon ( \tilde{x}^\varepsilon (\xi , t) , t ) = \frac{\rho_0 ( \xi )}{J^\varepsilon (\xi , t )}.
\end{equation*}
Using the above equalities, we prove \eqref{eq42}-\eqref{eq45}. Therefore the lemma follows.
\end{proof}

Let us now prove Theorem \ref{thm27}.
\begin{proof}[Proof of Theorem \ref{thm27}]

Assume that there are smooth functions $\tilde{y}$ and $z$ such that for $\xi \in \Omega (0)$ and $0 < t < T$,
\begin{equation*}
\begin{cases}
\frac{d}{d \varepsilon} \big|_{\varepsilon = 0} \tilde{x}^\varepsilon (\xi , t) = \tilde{y} (\xi , t),\\
z (\tilde{x} (\xi , t ) , t ) = \tilde{y} (\xi , t). 
\end{cases}
\end{equation*}
We now prove that
\begin{equation}
\tilde{y} ( \xi , 0 ) = 0. \label{eq46}
\end{equation}
Since $\tilde{x} (\xi , 0) = \tilde{x}^\varepsilon (\xi, 0) = \xi$, we see that
\begin{equation*}
\tilde{x}^\varepsilon (\xi , 0 ) - \tilde{x}(\xi , 0) = 0.
\end{equation*}
This gives $\tilde{y} (\xi , 0) = 0$.
\noindent Since
\begin{equation*}
\int_0^T \int_{\Omega^\varepsilon (t)} \frac{1}{2} \rho^\varepsilon | v^\varepsilon |^2 { \ } d x d t = \int_0^T \int_{\Omega (0)} \frac{1}{2} \rho_0 ( \xi ) \tilde{x}^\varepsilon_t (\xi , t) \cdot \tilde{x}_t^\varepsilon ( \xi , t) { \ } d \xi  d t,
\end{equation*}
we use the integration by parts and \eqref{eq46} to see that
\begin{align*}
\frac{d}{ d \varepsilon} \bigg|_{\varepsilon = 0} \int_0^T \int_{\Omega^\varepsilon (t)} \frac{1}{2} \rho^\varepsilon | v^\varepsilon |^2 { \ } d x d t & = \int_0^T \int_{\Omega (0)} \rho_0 ( \xi ) \tilde{x}_t (\xi , t) \cdot \tilde{y}_t ( \xi , t) { \ } d \xi  d t\\
& = \int_0^T \int_{\Omega (0)} \rho_0 ( \xi ) v ( \tilde{x} (\xi , t), t ) \cdot \tilde{y}_t ( \xi , t) { \ } d \xi  d t\\
& = - \int_0^T \int_{\Omega (0)} \rho_0 ( \xi ) D_t v \cdot \tilde{y} ( \xi , t) d \xi d t\\
& = - \int_0^T \int_{\Omega (t)} \{ \rho D_t  v \} (x,t) \cdot z (x,t){ \ }d x d t.
\end{align*}
Therefore Theorem \ref{thm27} is proved.
\end{proof}

\section{Variation of the Velocity to Dissipation Energies and Works}\label{sect5}

In this section we prove Theorems \ref{thm22} and \ref{thm23}. To this end, we prepare the following lemma.

\begin{lemma}[Variations of the velocity to dissipation energies]\label{lem51}
 For smooth functions $V = V (x,t) = { }^t (V_1 , V_2 , V_3)$ and $f = f ( x ,t )$,
\begin{align*}
E_{D_1} [V] (t) & : = - \int_{ \Omega (t)} \frac{1}{2} e_1 ( |D_+ ( V) |^2 ) { \ }d x,\\
E_{D_2} [V] (t) & : = - \int_{ \Omega (t)} \frac{1}{2} e_2 ( | {\rm{div}} V |^2 ) { \ }d x,\\
E_{D_3} [V] (t) & : = - \int_{ \Omega (t)} \frac{1}{2} e_3 ( |D_- ( V) |^2 ) { \ }d x,\\
E_{D_4} [f] (t)  & := - \int_{\Omega (t)} \frac{1}{2} e_4 (| {\rm{grad}} f |^2) { \ } d x.
\end{align*}
Then for every $\varphi = { }^t ( \varphi_1 , \varphi_2 , \varphi_3 ) \in C_0^\infty (\Omega (t))$ and $\phi \in C_0^\infty ( \Omega (t))$,
\begin{align*}
\frac{d}{d \varepsilon} \bigg|_{\varepsilon = 0} E_{D_1} [v + \varepsilon \varphi ] (t) & = \int_{\Omega (t)} {\rm{div}} \{ e_1' ( |D_+ ( v) |^2 ) D_+ (v) \} \cdot \varphi { \ } d x,\\
\frac{d}{d \varepsilon} \bigg|_{\varepsilon = 0} E_{D_2} [v + \varepsilon \varphi ] (t) & = \int_{\Omega (t)} {\rm{div}} \{ e_2' ( | {\rm{div}} v  |^2 ) ({\rm{div}} v )I_{ 3 \times 3} \} \cdot \varphi { \ } d x,\\
\frac{d}{d \varepsilon} \bigg|_{\varepsilon = 0} E_{D_3} [v + \varepsilon \varphi ] (t) & = \int_{\Omega (t)} {\rm{div}} \{ e_3' ( |D_- ( v) |^2 ) D_- (v) \} \cdot \varphi { \ } d x,\\
\frac{d}{d \varepsilon} \bigg|_{\varepsilon = 0} E_{D_4} [ \theta + \varepsilon \phi ] (t) & = \int_{\Omega (t)} {\rm{div}} \{ e_4' ( | {\rm{grad}} \theta |^2 ) {\rm{grad} \theta} \} \cdot \phi { \ } d x. 
\end{align*}
\end{lemma}

\begin{proof}[Proof of Lemma \ref{lem51}]
A direct calculation gives
\begin{align*}
\frac{d }{ d \varepsilon } \bigg|_{\varepsilon = 0} E_{D_1} [ v + \varepsilon \varphi ] & = - \int_{\Omega (t) } e'_1 (| D_+ ( v ) |^2 ) D_+ (v) : D_+ (\varphi ) { \ } d x,\\
\frac{d }{ d \varepsilon } \bigg|_{\varepsilon = 0} E_{D_2} [ v + \varepsilon \varphi ] & = - \int_{\Omega (t) } e'_2 (| {\rm{div}} v |^2 ) ({\rm{div}} v) ({\rm{div}} \varphi ) { \ } d x,\\
\frac{d }{ d \varepsilon } \bigg|_{\varepsilon = 0} E_{D_3} [ v + \varepsilon \varphi ] & = - \int_{\Omega (t) } e'_3 (| D_- ( v ) |^2 ) D_- (v) : D_- (\varphi ) { \ } d x,\\
\frac{d }{ d \varepsilon } \bigg|_{\varepsilon = 0} E_{D_4} [ \theta + \varepsilon \phi ] & = - \int_{\Omega (t) } e'_4 (| {\rm{grad}} \theta |^2 ) {\rm{grad}} \theta \cdot {\rm{grad}} \phi { \ } d x.
\end{align*}
Using integration by parts, we check that
\begin{align*}
- \int_{\Omega (t) } e'_1 (| D_+ ( v ) |^2 ) D_+ (v) : D_+ (\varphi ) { \ } d x & = \int_{\Omega (t) } {\rm{div}} \{ e'_1 (| D_+ ( v ) |^2 ) D_+ (v) \} \cdot \varphi  { \ } d x,\\
- \int_{\Omega (t) } e'_2 (| {\rm{div}} v |^2 ) ({\rm{div}} v) ({\rm{div}} \varphi ) { \ } d x & = \int_{\Omega (t) } {\rm{div}} \{ e'_2 (| {\rm{div}} v |^2 ) ({\rm{div}} v) I_{3 \times 3} \}  \cdot \varphi  { \ } d x,\\
- \int_{\Omega (t) } e'_3 (| D_- ( v ) |^2 ) D_- (v) : D_- (\varphi ) { \ } d x & = \int_{\Omega (t) } {\rm{div}} \{ e'_3 (| D_- ( v ) |^2 ) D_- (v) \}  \cdot \varphi  { \ } d x ,\\
- \int_{\Omega (t) } e'_4 (| {\rm{grad}} \theta |^2 ) {\rm{grad}} \theta \cdot {\rm{grad}} \phi { \ } d x & = \int_{\Omega (t) } {\rm{div}} \{ e'_4 (| {\rm{grad}} \theta |^2 ) {\rm{grad}} \theta \} \phi { \ } d x .
\end{align*}
Therefore Lemma \ref{lem51} is proved.
\end{proof}

\begin{proof}[Proof of Theorems \ref{thm22} and \ref{thm23}]
Using integration by parts, Lemma \ref{lem51}, and Proposition \ref{prop25}, we prove Theorems \ref{thm22} and \ref{thm23}.
\end{proof}

\section{Energetic Variational Approaches for Non-Newtonian Fluid Systems}\label{sect6}

In this section we make several mathematical models of non-Newtonian fluid flow. In subsection \ref{subsec61} we apply an energetic variational approach and the first law of thermodynamics to derive the generalized compressible non-Newtonian fluid system \eqref{eq11}. In subsection \ref{subsec62} we study the enthalpy, entropy, free energy, conservative forms, and conservation laws of the system \eqref{eq11}. In subsection \ref{subsec63} we make use of an energetic variational approach and Proposition \ref{prop25} to derive the generalized incompressible non-Newtonian fluid system.

\subsection{Energetic Variational Approaches for Compressible Fluid System}\label{subsec61}

Let us apply our energetic variational approaches to derive the generalized compressible non-Newtonian fluid system. We assume that $\Omega (t)$ is flowed by the velocity $v$. We set the energy densities for non-Newtonian as in Assumption \ref{ass11}. Based on Proposition \ref{prop26} we admit
\begin{equation}
\partial_t \rho + ( v , \nabla ) \rho  + ({\rm{div}} v ) \rho =0 . \label{eq61}
\end{equation}

We first derive the momentum equation of our compressible fluid system. Set
\begin{equation*}
S (v , \sigma ) = e'_1 ( |D_+ ( v) |^2 ) D_+ (v) + e'_2 (|{\rm{div}} v |^2 ) ({\rm{div}} v) I_{3}+ e'_3 ( | D_- (v) |^2) D_- (v) - \sigma I_{3} . 
\end{equation*}
From Theorems \ref{thm23} and \ref{thm27}, we have the following forces:
\begin{align*}
\frac{\delta E_{D + W}}{\delta v} & = {\rm{div}} S ( v , \sigma ) + \rho F,\\
\frac{\delta A}{\delta \tilde{x}} & = \rho D_t v .
\end{align*}
We assume the following energetic variational principle:
\begin{equation*}
\frac{\delta A}{\delta \tilde{x}} = \frac{\delta E_{D + W}}{\delta v} 
\end{equation*}
to have 
\begin{equation*}
\rho D_t v = {\rm{div}} S ( v , \sigma ) + \rho F.
\end{equation*}
This is equivalent to
\begin{equation}
\rho D_t v + {\rm{grad}} \sigma = {\rm{div}} S ( v , 0 ) + \rho F. \label{eq62}
\end{equation}

Secondly, we apply the first law of thermodynamics to derive the dominant equation for the internal energy. To this end, we now consider both the energy dissipation due the viscosities and the work done by the pressure. Multiplying the system \eqref{eq62} by $v$, then using integration by parts and Gauss's divergence theorem, we observe that for $0 < t_1 < t_2 <T$,
\begin{multline*}
\int_{\Omega ( t_2 )} \frac{1}{2} \rho | v |^2 { \ }d x + \int_{t_1}^{t_2} \int_{\Omega ( \tau )} (\tilde{e}_D - ({\rm{div}} v ) \sigma ) { \ } d x d \tau \\
= \int_{\Omega (t_1)} \frac{1}{2} \rho | v |^2 { \ } d x + \int_{t_1}^{t_2} \int_{\Omega ( \tau )} \rho F \cdot v { \ } d x d \tau + \int_{t_1}^{t_2} \int_{\partial \Omega ( \tau )} BC { \ }d S d \tau .
\end{multline*}
Here
\begin{equation*}
\tilde{e}_D =  e'_1 (| D_+ ( v) |^2) |D_+(v)|^2 + e'_2 (| {\rm{div}} v |^2) | {\rm{div}} v |^2 ) + e'_3 (| D_- (v)|^2 ) |D_-(v)|^2,
\end{equation*}
\begin{multline*}
BC  = [e'_1 (| D_+ ( v) |^2) D_+(v) + e'_2 (| {\rm{div}} v |^2) ({\rm{div}} v ) I_{3 \times 3} \\
+ e'_3 (| D_- (v)|^2 ) D_-(v) - \sigma I_{3 \times 3} ] n |_{\partial \Omega (\tau )}. 
\end{multline*}
From now on we assume that $BC \equiv 0$. Then we have an energy equality:
\begin{multline*}
\int_{\Omega ( t_2 )} \frac{1}{2} \rho | v |^2 { \ }d x + \int_{t_1}^{t_2} \int_{\Omega ( \tau )} (\tilde{e}_D - ({\rm{div}} v ) \sigma ) { \ } d x d \tau \\
= \int_{\Omega (t_1)} \frac{1}{2} \rho | v |^2 { \ } d x +  \int_{t_1}^{t_2} \int_{\Omega ( \tau )} \rho F \cdot v { \ } d x d \tau .
\end{multline*}
This shows that $\tilde{e}_D - ({\rm{div}} v ) \sigma $ is the dissipation energy and the work done by the pressure on our compressible fluid system. From Theorem \ref{thm24} we have the following forces:
\begin{align}
\frac{\delta E_{TD}}{ \delta \theta } & = {\rm{div}} \{ e'_4 ( |{\rm{grad}}  \theta |^2 ) {\rm{grad}} \theta \}, \label{eq63}\\
\frac{\delta E_{GD}}{ \delta C} & = {\rm{div}} \{ e'_5 ( |{\rm{grad}}  C |^2 ) {\rm{grad}} C \}.\label{eq64}
\end{align}
We apply the first law of thermodynamics to derive
\begin{equation}
\rho D_t e + ({\rm{div}} v ) \sigma = {\rm{div}} \{ e'_4 (|{\rm{grad}}  \theta |^2 ) {\rm{grad}} \theta \} + \tilde{e}_D . \label{eq65}
\end{equation}
More precisely, we assume that for every $U (t) \subset \Omega (t)$ flowed by the velocity $v$,
\begin{equation*}
\frac{d}{d t} \int_{U (t)} \rho e { \ } d x = \int_{U (t)} \left\{ \frac{\delta E_{TD}}{ \delta \theta }  + \tilde{e}_D - ({\rm{div}} v ) \sigma \right\} { \ } d x.
\end{equation*}
Then we have \eqref{eq65}.

Finally, we derive the generalized diffusion system. We assume that the change of rate of the concentration $C$ equals to the force derived from a variation of the energy dissipation due to general diffusion, that is, for every $U (t) \subset \Omega (t)$ flowed by the velocity $v$, assume that
\begin{equation*}
\frac{d}{d t} \int_{U (t)} C { \ } d x = \int_{U (t)} \frac{\delta E_{GD}}{ \delta C } { \ } d x.
\end{equation*}
Applying the Reynold transport theorem, we have
\begin{equation}
D_t C + ({\rm{div}} v ) C = {\rm{div}} \{ e_5' (|{\rm{grad}}  C |^2 ) {\rm{grad}} C \}. \label{eq66}
\end{equation}
Combining \eqref{eq61}, \eqref{eq62}, \eqref{eq65}, and \eqref{eq66}, we therefore have our compressible non-Newtonian fluid system \eqref{eq11}.

\subsection{On Generalized Compressible non-Newtonian Fluid System}\label{subsec62}{ \ }\\
Let us consider the generalized compressible non-Newtonian fluid system \eqref{eq11}. We first study the enthalpy, entropy, and free energy of our compressible fluid system. Then we consider the total energy, conservative forms, and conservation laws of the system \eqref{eq11}. We admit the system \eqref{eq11}.

Assume that $\rho$ and $\theta$ are positive functions. Set the enthalpy $h = h (x,t)$ as follows $h = e + \sigma / \rho$. Then
\begin{equation*}
\rho D_t h = {\rm{div}} \{ e'_4 (|{\rm{grad} } \theta |^2 ) {\rm{grad}} \theta \} + \tilde{e}_D + D_t \sigma \text{ in } \Omega_T .
\end{equation*}
Assume that the entropy $s = s ( x ,t )$ satisfies the Gibbs condition: 
\begin{equation*}
D_t e = \theta D_t s - \sigma D_t \left( \frac{1}{\rho} \right) \text{ in } \Omega_T .
\end{equation*}
Then
\begin{equation*}
\theta \rho D_t s = {\rm{div}} \{ e'_4 (|{\rm{grad} } \theta |^2 ) {\rm{grad}} \theta \} + \tilde{e}_D \text{ in } \Omega_T. 
\end{equation*}
Set the Helmholtz free energy $e_F = e - \theta s$. An easy calculation gives
\begin{equation*}
\rho D_t e_F + s \rho D_t \theta - S (v , \sigma ) : ( D_+ (v) + D_- (v) ) = - \tilde{e}_D \text{ in } \Omega_T.
\end{equation*}
Therefore we have
\begin{equation}\label{eq67}
\begin{cases}
\rho D_t h = {\rm{div}} \{ e'_4 (|{\rm{grad} } \theta |^2 ) {\rm{grad}} \theta \} + \tilde{e}_D + D_t \sigma  &\text{ in } \Omega_T,\\
\theta \rho D_t s = {\rm{div}} \{ e'_4 (|{\rm{grad} } \theta |^2 ) {\rm{grad}} \theta \} + \tilde{e}_D &\text{ in } \Omega_T,\\
\rho D_t e_F + s \rho D_t \theta - S (v , \sigma ) : ( D_+ (v) + D_- (v) ) = - \tilde{e}_D &\text{ in } \Omega_T.
\end{cases}
\end{equation}

Next we consider conservative form of the system \eqref{eq11}. Set the total energy $e_A = \rho | v |^2 /2 + \rho e$. We easily check that the systems \eqref{eq11} and \eqref{eq67} satisfy \eqref{eq12} and \eqref{eq13}. 

Finally, we consider the conservation laws the system \eqref{eq11} to prove Theorem \ref{thm28}. We assume that $\Omega (t)$ is flowed by the velocity $v$. Assume that for each $0 < t < T$,
\begin{align}
e'_4 (|{\rm{grad}} \theta |^2 ) (n , \nabla ) \theta |_{\partial \Omega (t)} & = 0,\label{eq68}\\
e'_5 (|{\rm{grad}} C |^2 ) (n , \nabla ) C |_{\partial \Omega (t)} & = 0,\label{eq69}
\end{align}
and
\begin{multline}\label{eq610}
[e'_1 (| D_+ ( v) |^2) D_+(v) + e'_2 (| {\rm{div}} v |^2) ({\rm{div}} v ) I_{3 \times 3} \\
+ e'_3 (| D_- (v)|^2 ) D_-(v) - \sigma I_{3 \times 3} ] n |_{\partial \Omega (t)} = { }^t (0 , 0 , 0 ).
\end{multline}
Since
\begin{align*}
\frac{d}{ d t} \int_{\Omega (t)} \rho v { \ }d x & = \int_{\Omega (t)} \rho D_t v { \ } d x ,\\
\frac{d}{ d t} \int_{\Omega (t)} e_A { \ }d x & = \int_{\Omega (t)} \{ \rho D_t v \cdot v + \rho D_t e \} { \ } d x ,\\
\frac{d}{ d t} \int_{\Omega (t)} C { \ }d x & = \int_{\Omega (t)}  \{ D_t C + ({\rm{div}}v ) C \} { \ } d x ,
\end{align*}
we use \eqref{eq11}, integration by parts, and the boundary conditions \eqref{eq68}-\eqref{eq610} to find that for $0 < t_1 < t_2 < T$,
\begin{align*}
\int_{\Omega (t_2)} \rho v { \ } d x & = \int_{\Omega (t_1)} \rho v { \ } d x + \int_{t_1}^{t_2} \int_{\Omega ( \tau )} \rho F { \ }d x d \tau ,\\
\int_{\Omega (t_2)} e_A { \ } d x & = \int_{\Omega (t_1)} e_A { \ } d x + \int_{t_1}^{t_2} \int_{ \Omega ( \tau )} \rho F \cdot v { \ }d x d \tau,\\
\int_{\Omega (t_2)} C { \ } d x & = \int_{\Omega (t_1)} C { \ } d x.
\end{align*}
Therefore the assertion $(\mathrm{i})$ of Theorem \ref{thm28}. 

Next we show the assertion $(\mathrm{ii})$. Suppose that $e_3 ( r ) = \mu_3 r$ for some $\mu_3 \in \mathbb{R}$. Assume that for $0 < t <T$,
\begin{multline}\label{eq611}
[(e'_1 (| D_+ ( v) |^2) + \mu_3 )D_+(v) + (e'_2 (| {\rm{div}} v |^2) - \mu_3) ({\rm{div}} v ) I_{3} - \sigma I_{3} ] n |_{\partial \Omega (t)} \\
= { }^t (0 , 0 , 0 ).
\end{multline}
Applying a flow map and the Reynold transport theorem, we check that
\begin{align*}
\frac{d}{d t} \int_{\Omega (t)} x \times \rho v { \ }d x & = \frac{d}{d t} \int_{\Omega (0)} \tilde{x} ( \xi , t ) \times \rho_0 (\xi ) v (\tilde{x}(\xi , t ) , t ) { \ }d \xi\\
= &  \int_{\Omega (t)} x \times \rho D_t v { \ }d x\\
= & \int_{\Omega (t)} x \times ( {\rm{div}} S ( v , \sigma ) + \rho F ) { \ }d x. 
\end{align*}
Here we used the fact that $v \times v = 0$. From $e_3 (r) = \mu_3 r$, we find that
\begin{equation*}
{\rm{div}} \{ e_3' (|D_- (v)|^2) D_- (v)  \}  = {\rm{div}} \{ \mu_3 D_+ (v) - \mu_3 ({\rm{div}} v ) I_{3 \times 3} \} .
\end{equation*}
Since
\begin{equation*}
( e_1' (|D_+ (v)|^2) + \mu_3 ) D_+ (v) + ( e_2' (| {\rm{div}} v|^2 ) - \mu_3 ) ({\rm{div}}v ) I_{3 \times 3} - \sigma I_{3 \times 3}
\end{equation*}
is a symmetric matrix, we use the integration by parts and the assumption \eqref{eq611} to have
\begin{equation*}
\int_{\Omega (t)} x \times {\rm{div}} S ( v , \sigma ) { \ }d x = 0. 
\end{equation*}
Therefore we have
\begin{equation*}
\int_{\Omega (t_2)} x \times \rho v { \ } d x = \int_{\Omega (t_1)} x \times \rho v { \ } d x + \int_{t_1}^{t_2} \int_{\Omega ( \tau )} x \times \rho F { \ }d x d \tau .
\end{equation*}
Therefore Theorem \ref{thm28} is proved.

\subsection{Energetic Variational Approaches for Incompressible Fluid System}\label{subsec63}
Let us apply our energetic variational approaches to derive the generalized incompressible non-Newtonian fluid system. We assume that $\Omega (t)$ is flowed by the velocity $v$. We set the energy densities for non-Newtonian as in Assumption \ref{ass11}.

We first consider the continuity equation and incompressible condition of our non-Newtonian fluid system. From Proposition \ref{prop26} we see that if ${\rm{div}}v = 0$ then
\begin{equation*}
\frac{d }{d t} \int_{\Omega (t)} 1 { \ }d x = \int_{\Omega (t)} {\rm{div}} v { \ } d x = 0.
\end{equation*}
Integrating with respect to time, we find that for $0 < t_1 < t_2 < T$
\begin{equation*}
\text{Vol}(\Omega (t_2)) =\int_{\Omega (t_2)} 1 { \ } d x = \int_{\Omega (t_1)} 1 { \ } d x = \text{Vol} (\Omega (t_1)) .
\end{equation*}
We also see that if ${\rm{div}} v =0$ then for each $U(t) \subset \Omega (t)$ flowed by the velocity $v$,
\begin{equation*}
\frac{d }{d t} \int_{U (t)} 1 { \ }d x = 0.
\end{equation*}
This means the local volume preservation. Therefore we admit the following incompressible condition and continuity equation:
\begin{align}
{\rm{div}} v = 0 \label{eq612},\\
D_t \rho =0 \label{eq613}.
\end{align}

Secondly we derive the momentum equation of our incompressible fluid system. From Theorems \ref{thm23} and \ref{thm27}, we have the following forces:
\begin{align*}
\frac{\delta E_{D + W}}{\delta v} \bigg|_{{\rm{div}} \varphi = 0} & = {\rm{div}} \{ e_1' ( |D_+ ( v) |^2 ) D_+ (v) + e_3' ( | D_- (v) |^2) D_- (v) \} - {\rm{grad}} \sigma + \rho F,\\
\frac{\delta A}{\delta \tilde{x}} & = \rho D_t v .
\end{align*}
We assume the following energetic variational principle:
\begin{equation*}
\frac{\delta A}{\delta \tilde{x}} = \frac{\delta E_{D + W}}{\delta v} 
\end{equation*}
to have 
\begin{equation}
\rho D_t v + {\rm{grad}} \sigma = {\rm{div}} \{ e'_1 (|D_+ (v) |^2 ) D_+ (v) + e'_3 (| D_- ( v ) |^2 ) D_- (v) \} + \rho F. \label{eq614}
\end{equation}
This is the momentum equation of our incompressible fluid system. Note that we may use 
\begin{equation}\label{eq615}
\frac{\delta A}{\delta \tilde{x}} \bigg|_{{\rm{div}} z = 0} = \rho D_t v  + {\rm{grad}} \sigma .
\end{equation}
See Proposition \ref{prop82} for details.

Thirdly, we derive the generalized heat and diffusion equations. Theorem \ref{thm24} gives \eqref{eq63} and \eqref{eq64}. We assume that the change of rate of the heat energy equals to the force derived from a variation of the energy dissipation due to thermal diffusion, that is, for every $U (t) \subset \Omega (t)$ flowed by the velocity $v$, assume that
\begin{equation*}
\frac{d}{d t} \int_{U (t)} \rho \theta { \ } d x = \int_{U (t)} \frac{\delta E_{TD}}{ \delta \theta } { \ } d x.
\end{equation*}
Applying the Reynold transport theorem, we have
\begin{equation}
\rho D_t \theta = {\rm{div}} \{ e_4' (|{\rm{grad}}  \theta |^2 ) {\rm{grad}} \theta \}. \label{eq616}
\end{equation}
Similarly, we assume that the change of rate of the concentration equals to the force derived from a variation of the energy dissipation due to general diffusion to derive
\begin{equation}
D_t C = {\rm{div}} \{ e_5' (|{\rm{grad}}  C |^2 ) {\rm{grad}} C \}. \label{eq617}
\end{equation}
Combining \eqref{eq612}-\eqref{eq617}, we have our incompressible non-Newtonian fluid system.

\section{Appendix (I): Derivation of Strain Rate Tensors and Fluxes}

In the Appendix (I) we introduce a method for deriving strain rate tensors and fluxes from our energy densities.

\begin{proposition}[Derivation of strain rate tensors and fluxes]\label{prop71}{ \ }\\
$(\mathrm{i})$ For $\vartheta = ( \vartheta_{i j})_{3 \times 3} \in M_{3\times 3} (\mathbb{R})$,
\begin{align*}
\mathbb{D} ( \vartheta ) := \vartheta, { \ }\mathbb{D}_+ (\vartheta ) := \frac{1}{2} \{ \mathbb{D}(\vartheta ) + { }^t (\mathbb{D} (\vartheta ) ) \},{ \ }\mathbb{D}_- (\vartheta ) := \frac{1}{2} \{ \mathbb{D}(\vartheta ) - { }^t (\mathbb{D} (\vartheta ) ) \}.
\end{align*}
Set
\begin{align*}
\mathcal{E}_1 [\vartheta ]  = - \frac{1}{2} e_1 ( | \mathbb{D}_+ (\vartheta ) |^2 ),{ \ }\mathcal{E}_2 [\vartheta ]  = - \frac{1}{2} e_2 ( | {\rm{Tr}} \mathbb{D} (\vartheta ) |^2 ),{ \ }\mathcal{E}_3 [\vartheta ]  = - \frac{1}{2} e_3 ( | \mathbb{D}_- (\vartheta ) |^2 ).
\end{align*}
Then
\begin{equation*}
\begin{pmatrix}
\frac{\partial \mathcal{E}_k}{\partial \vartheta_{11}} & \frac{\partial \mathcal{E}_k}{\partial \vartheta_{12}} & \frac{\partial \mathcal{E}_k}{\partial \vartheta_{13}} \\
\frac{\partial \mathcal{E}_k}{\partial \vartheta_{21}} & \frac{\partial \mathcal{E}_k}{\partial \vartheta_{22}} & \frac{\partial \mathcal{E}_k}{\partial \vartheta_{23}} \\
\frac{\partial \mathcal{E}_k}{\partial \vartheta_{31}} & \frac{\partial \mathcal{E}_k}{\partial \vartheta_{32}} & \frac{\partial \mathcal{E}_k}{\partial \vartheta_{33}} 
\end{pmatrix}\Bigg|_{(\vartheta_{i j} = \partial_j v_i)} = 
\begin{cases}
- e_1' (|D_+ (v) |^2) D_+ ( v ) & k =1,\\
- e_2' (|{\rm{div}} v |^2) ( {\rm{div}} v ) I_{3 \times 3} & k =2,\\
- e_3' (| D_- (v) |^2 ) D_- (v) & k =3.
\end{cases}
\end{equation*}
\noindent $(\mathrm{ii})$ For $\vartheta = { }^t ( \vartheta_1 , \vartheta_2 , \vartheta_3 ) \in \mathbb{R}^3$,
\begin{align*}
\mathcal{E}_4 [ \vartheta ] = - \frac{1}{2} e_4 (| \vartheta |^2),{ \ }\mathcal{E}_5 [ \vartheta ] = - \frac{1}{2} e_5 (| \vartheta |^2). 
\end{align*}
Then
\begin{align*}
& { }^t \left( \frac{\partial \mathcal{E}_4 }{\partial \vartheta_1}, \frac{\partial \mathcal{E}_4 }{\partial \vartheta_2 },  \frac{\partial \mathcal{E}_4 }{\partial \vartheta_3}  \right) \bigg|_{\vartheta_j = \partial_j \theta } = - e_4' (|{\rm{grad} }\theta |^2) {\rm{grad}} \theta,\\
& { }^t \left( \frac{\partial \mathcal{E}_5 }{\partial \vartheta_1}, \frac{\partial \mathcal{E}_5 }{\partial \vartheta_2 },  \frac{\partial \mathcal{E}_5 }{\partial \vartheta_3}  \right) \bigg|_{\vartheta_j = \partial_j C } = - e_5' (|{\rm{grad} }C |^2) {\rm{grad}} C .
\end{align*}
\end{proposition}
The proof of Proposition \ref{prop81} can be proved by a simple calculation.

\section{Appendix (II): On Derivation of Inviscid Fluid Systems}\label{sect8}

In the Appendix (II), we introduce energetic variational approaches for the compressible and incompressible inviscid fluid systems. We often call the inviscid fluid systems the \emph{Euler systems}. In this paper we deal with the compressible barotropic fluid.

\subsection{Compressible Inviscid Fluid System}\label{subsec81}{ \ }\\
Let us derive the compressible inviscid fluid system. We admit $\rho D_t v + ({\rm{div}} v ) \rho = 0$ form Proposition \ref{prop26}. Now we derive the momentum equation of the compressible inviscid fluid system. We make use of a chemical potential to derive the pressure of the barotropic fluid.
\begin{proposition}\label{prop81}
Let $p $ be a $C^1$-function. Under the same assumption of Theorem \ref{thm27}, we set
\begin{equation*}
A_B [ \tilde{x}^\varepsilon ] = - \int_0^T \int_{\Omega^\varepsilon (t)} \left( \frac{1}{2} \rho^\varepsilon | v^\varepsilon |^2 - p (\rho^\varepsilon ) \right) { \ }d x d t.
\end{equation*}
Then
\begin{equation*}
\frac{d}{ d \varepsilon } \bigg|_{\varepsilon = 0}A_B [ \tilde{x}^\varepsilon ] = \int_0^T \int_{\Omega (t)} \{ \rho D_t v + {\rm{grad}} \mathfrak{p} \} ( x,t) \cdot z (x,t) { \ }d x d t,\end{equation*}
where $\mathfrak{p} = \mathfrak{p} ( \rho ) = \rho p' (\rho) - p (\rho)$ and $z$ is a variation of $\tilde{x}^\varepsilon$. 
\end{proposition}
\noindent Since we can prove Proposition \ref{prop81} by an argument similar to that in Section \ref{sect4}, the proof of above Proposition is left to the reader. See Forster \cite[Example 17]{For13} for another approach for Proposition \ref{prop81}.

From Proposition \ref{prop81}, we have
\begin{equation*}
\frac{ \delta A_B}{ \delta \tilde{x}} = \rho D_t v + {\rm{grad}} \mathfrak{p}.
\end{equation*}
This is the momentum equation of the compressible inviscid fluid system. Therefore we have the following compressible inviscid fluid system:
\begin{equation*}
\begin{cases}
D_t \rho + ({\rm{div}} v ) \rho = 0& \text{ in } \Omega_T,\\ 
\rho D_t v + {\rm{grad}} \mathfrak{p} = 0& \text{ in } \Omega_T,\\
\mathfrak{p} = \mathfrak{p} (\rho ) = \rho p' ( \rho ) - p ( \rho )& \text{ in } \Omega_T, 
\end{cases}
\end{equation*}
where $p$ is a $C^1$-function. We often call $p (\rho )$ a \emph{chemical potential}. See Koba \cite{K17} for mathematical derivation of the compressible inviscid fluid system on an evolving surface.

\subsection{Incompressible Inviscid Fluid System}\label{subsec82}{ \ }\\
Let us derive the incompressible inviscid fluid system. By the arguments in subsection \ref{subsec63}, we admit $D_t \rho = 0$ and ${\rm{div}}v = 0$. Based on Theorem \ref{thm27}, we apply Proposition \ref{prop25} to deduce following proposition:
\begin{proposition}\label{prop82}
Fix $t \in (0, T)$. For every $z \in [C_0^\infty ( \Omega (t))]^3$ satisfying ${\rm{div}} z = 0$, assume that
\begin{equation*}
\int_{\Omega (t)} \{ \rho D_t v \} (x,t) \cdot z(x,t) { \ }d x = 0.
\end{equation*}
Then there is a function $\sigma \in C^1 ( \Omega (t)) $ such that
\begin{equation*}
\rho D_t v = {\rm{grad}} \sigma .
\end{equation*}
\end{proposition}
\noindent The above proposition gives the momentum equation of the incompressible inviscid fluid system. From Theorem \ref{thm27} and Proposition \ref{prop82}, we obtain \eqref{eq615}. Consequently, we have the following incompressible inviscid fluid system:
\begin{equation*}
\begin{cases}
D_t \rho = 0 & \text{ in } \Omega_T,\\
{\rm{div}} v = 0& \text{ in } \Omega_T,\\ 
\rho D_t v + {\rm{grad}} \sigma = 0& \text{ in } \Omega_T.
\end{cases}
\end{equation*}
See Arnol'd \cite{Arn97} and Koba-Liu-Giga \cite{KLG17} for mathematical derivations of incompressible inviscid fluid systems on a manifold and an evolving surface, respectively.

\begin{flushleft}
{\bf{Acknowledgments:}}{ \ }\\
The authors would like to thank Professor Takayuki Kobayashi for his support.
\end{flushleft}

\begin{flushleft}
{\bf{Conflict of interest :}} The authors declare that they have no conflict of interest.
\end{flushleft}

\end{document}